\RequirePackage{snapshot}
\documentclass[twoside]{article}
\usepackage{ifxetex}
\ifxetex
	\usepackage[no-math]{fontspec}
\else
	\pdfoutput=1
\fi
\usepackage[margin=1in]{geometry}
\usepackage{amsmath,amssymb,mathtools,graphicx,subcaption,booktabs,array,multirow,multicol}
\usepackage[nounderscore]{syntax}
\usepackage{listings,color,verbatim,url,textcomp}
\usepackage{mycommands,mytheorems,formal_methods_commands}
\usepackage{algorithm,algpseudocode}

\pdfmapfile{+ltlfonts.map}
\usepackage{pf2}

\def\tla@while{-\hspace{-.18em}\triangleright}
\def\tlawhilep{\stackrel{\mbox{\raisebox{-.3em}[0pt][0pt]{$\scriptscriptstyle+\;\,$}}}{\tla@while}}

\def\strictarrow{\stackrel{\mbox{\raisebox{-.3em}[0pt][0pt]{$\scriptscriptstyle\mathrm{sr}\;\,$}}}{\tla@while}}

\graphicspath{{./img/}{./img/program_graph/}}

\algtext*{EndIf}
\algtext*{EndFor}
\algtext*{EndWhile}
\algtext*{EndProcedure}

\newenvironment{pf2proof}{
\renewenvironment{proof}{
\begin{pfproof}
}{
\end{pfproof}
}
}{}

\newcommand{\cl}[2][]{\mathcal{C}_{#1}\big(#2\big)}

\renewcommand{\include}{\input}
\begin{document}
	\title{Symbolic construction of GR(1) contracts for synchronous systems with full information}
\author{Ioannis Filippidis \hspace{1cm} Richard M. Murray\\[0.4cm]
\texttt{\{ifilippi,murray\}{@}caltech.edu}\\[0.2cm]
Control and Dynamical Systems\\
California Institute of Technology}
\maketitle

\begin{abstract}

This work proposes a symbolic algorithm for the construction of assume-guarantee specifications that allow multiple agents to cooperate.
Each agent is assigned goals expressed in a fragment of linear temporal logic known as generalized reactivity of rank 1 (GR(1)).
These goals may be unrealizable, unless additional assumptions are made by each agent about the behavior of the other agents.
The proposed algorithm constructs weakly fair assumptions for each agent, to ensure that they can cooperate successfully.
A necessary requirement is that the given goals be cooperatively satisfiable.
We prove that there exist games for which the GR(1) fragment with liveness properties over states is not sufficient to ensure realizability from any state in the cooperatively winning set.
The obstruction is due to circular dependencies of liveness goals.
To prevent circularity, we introduce nested games as a formalism to express specifications with conditional assumptions.
The algorithm is symbolic, with fixpoint structure similar to the GR(1) synthesis algorithm, implying time complexity polynomial in the number of states, and linear in the number of recurrence goals.

\end{abstract}

\tableofcontents

	\section{Introduction}

The design and construction of a large system relies on the ability to divide the problem into smaller ones.
Each subproblem involves a subset of the system, and may itself be refined further into smaller problems.
The subsystems that result from the smaller problems are considered as modules of the larger system.
In many cases, the modules interact with each other, either physically, or as software, or both.
For this reason, the interaction between modules needs to be constrained, in order to ensure that the modules can perform their operation as intended.
For example, if we consider the fridge and a power plug as modules of a house, then the fridge can only preserve food provided that the plug provides electric power uninterruptedly.
In many cases, modularization is a necessity, imposed by the topology of the design, because the system comprises of elements distributed over space.
These elements control some local part of the system, but they need to communicate, in order to coordinate.

Among the benefits of modularization are the division of a complex problem into smaller ones that are computationally cheaper to solve, the localization of reasoning, which focuses the designer's attention and reduces the danger of errors, and the ability to assign the design of each subsystem to a different entity, for example a contractor specializing in that type of system.
In addition, a well-defined description of each individual component enables re-using the same design in a different context, where such a component is needed.
This leads to the possibility of interfacing off-the-shelf components, based on their interface description, thus reducing the need for case-by-case design and production.

In order to describe a module and its interaction with other modules, and their environment, it is necessary to represent them.
A representation can range from an informal textual description, to a mathematically defined notation, with fixed syntax and semantics.
The latter is desirable, because it is not ambiguous, and it enables automation of checking whether a candidate solution satisfies the requirements.
Such a formal representation is usually called a {\em specification}.

Proving that a system will behave as intended, insofar as this is captured by a specification, is a major objective in systems that are critical for the safety of humans, or have a very high cost.
These include aircraft, especially airliners, spacecraft, which is a major investment and missions are, in many cases, unique and not to be repeated, automotive subsystems, nuclear power plant controllers, and several other application areas.

We can distinguish two broad problems, at different phases of system design.
The first one asks for producing formal specifications that describe the modules, with detail sufficient to allow for automated synthesis.
The second problem asks for constructing an implementation of each module, and assembling the results into a complete composite system.
The first problem comprises the modularization and specification step, whereas the second is the construction phase.

The specification of a system can be implemented by humans, or constructed by an algorithm.
The latter approach is known as (automated) {\em synthesis}, and relies on notions from the theory of games \cite{Thomas08npgi}.
Synthesis has attracted considerable interest in the past two decades, and advances both in theory and implementation have been made, as described in the following sections.
In this work, we are interested in algorithmic synthesis, for both phases of system design.
In particular, we aim at automatically modularizing a design that has been partially specified by a human.
In other words, humans give as input a formal description about what each module is expected to accomplish.
Note that this step is necessary, in one form or another, because the algorithm cannot know what the modules are intended for.
We consider these as the primitive specifications, that are given by a human, and will be completed algorithmically.
These specifications may be insufficient for obtaining a coherent system, but describe the goals, and provide the starting point for an algorithmic approach to complete the specifications, and then construct implementations.
The automated modularization step involves {\em completing} the specifications, by adding more detail, in a way that ensures that there exist components satisfying the primitive specifications.
Regarding the synthesis phase, we are interested in efficient and scalable synthesis algorithms that can handle specifications with many goals.

Clearly, the formal description of all the details in a given implementation is itself a specification.
However, fixing a particular implementation is usually much more restrictive than needed.
It is desired to describe only what is necessary of a particular module, and leave the internal details of its exact operation to be decided by the implementor of that module.
The difference between an implementation, and a less restrictive specification is {\em quantification}.
A specification contains {\em existential} quantification, if it asks for some type of behavior contained in a given set, but does not describe a particular instance of that behavior.
Another term that is commonly used to characterize this quality of a specification is {\em declarative}.

This motivates regarding synthesis as a {\em compiler} activity.
In analogy with declarative programming languages like Haskell, a declarative specification is intended to leave unconstrained the exact imperative details of {\em how} the implementation will behave, step-by-step.
A synthesizer from declarative specifications {\em compiles} them into an implementation that operates in time by reading environment inputs, and writing outputs.
Reading and writing are used here in a broad sense, meaning interaction that may involve mechanical or hydraulic forces.
A distinction between conventional declarative languages, and module synthesis is that the latter produces components that continue to interact with their environment without ever terminating, also known as {\em reactive} systems.
For example, the computer that controls a nuclear reactor is not intended to terminate and produce some result, under normal operation.
This is in contrast to a matrix multiplication program.

\section{Modular design by contract}
\label{sec:design-by-contract}

There have been several approaches to the modularization of systems.
The design of each module becomes simpler, because it involves fewer elements, as counted, for example, by the number of variables used to represent it.
However, the challenge is shifted, from designing a monolithic system, to putting together the pieces.
In this report, we consider the problem of interfacing the modules.

Our approach constructs specifications that are partitioned into assumptions about the behavior of the world outside a component, and requirements that the component guarantees, provided its assumptions hold.
This is known as assumption-commitment, or rely-guarantee paradigm for describing behaviors.

The assumption-commitment paradigm about reactive systems is an evolved instance of reasoning about conditions before, and after, a terminating behavior.
A formalism for reasoning using triples of a {\em precondition}, a program, and a {\em postcondition} was introduced by Hoare \cite{Hoare69cacm}, following the work of Floyd \cite{Floyd67sam} on proving properties of elements in a flowchart, based on ideas by Perlis and Gorn \cite{Schneider97}.

Hoare's logic applies to terminating programs.
However, many systems are not intended to terminate, but instead continue to operate, by reacting to their environment \cite{Pnueli89popl}.
Francez and Pnueli \cite{Francez78ai} introduced a first generalization of Hoare-style reasoning to cyclic programs.
They also considered concurrent programs.
Their formalism uses explicit mention of time, and is structured into pairs of assumptions and commitments.

Lamport \cite{Lampor83toplas} observed that such a style of specification is essential to reason about complex systems in a modular way.
Lamport and Schneider \cite{Lampor80ai,Lamport84toplas} introduced, and related to previous approaches, what they called {\em generalized Hoare logic}.
This is a formalism for reasoning with pre- and post-conditions, in order to prove program invariants.
Misra and Chandy introduced the rely-guarantee approach for safety properties of distributed systems \cite{Misra81tse}, still for safety properties.
All properties up to this point were safety, and not expressed in temporal logic \cite{Pnueli77}.
Two developments followed, and the work presented here is based on them.

The first was Lamport's introduction of {\em proof lattices} \cite{Lamport77tse}.
A proof lattice is a finite rooted directed acyclic graph, labeled with assertions.
If $u$ is a node labeled with property $U$, and $v, w$ are its successors, labeled with properties $V, W$, then if $U$ holds at any time, eventually either $V$ or $W$ will hold.
In temporal logic, this can be expressed as $\always (U \rightarrow \eventually (V \vee W))$.
Owicki and Lamport \cite{Owicki82toplas} revised the proof lattice approach, by labeling nodes with {\em temporal} properties, instead of atemporal ones (“immediate assertions”).

The second development was the expression by Pnueli \cite{Pnueli85lmcs} of assume-guarantee pairs in temporal logic, i.e., without reference to an explicit {\em time} variable.
In addition, Pnueli proposed a proof method for {\em liveness} properties, which is based on well-founded induction.
This method can be understood as starting with some temporal premises for each component, and iteratively tightening these properties into consequents that are added to the collection of available premises, for the purpose of deriving further consequents.
This method enables proving liveness properties of modular systems.
Informally, the requirement of well-foundedness allows using as premises only properties from an earlier stage of the deductive process.
This prevents circular existential reasoning about the future, i.e., circular dependencies of liveness properties.
As a simple example \cite{Abadi94podc}, consider Alice and Bob.
Alice promises that, if she sees $b$, then she will do $a$ at {\em some} time in the future.
Reciprocally, Bob promises to eventually do $b$, after he sees $a$.
As linear temporal logic (LTL) formulae, these read $\always(b \rightarrow \nxt \eventually a)$ for Alice, and with $a, b$ swapped, for Bob.
If both Alice and Bob default to not doing any of $a$ or $b$, then they both satisfy their specifications.
This problem arises, because existential quantification in {\em future}\footnote{
	Compare with existential quantification in {\em past} time, as is the case in the past fragment of LTL.
	This causes no problems, because it concerns things past.
}
time allows simultaneous antecedent failure.
Otherwise, if Bob was {\em required} to do $b$ for the first time, then Alice would have to do $a$, then Bob do $b$ again, etc.

Compositional approaches to verification have treated the issue of circularity by using the description of the model under verification as a vehicle for carrying out the proof.
In other words, the immediate behavior of the model, as captured by its transition relation, should constrain the system sufficiently much, so as to enable deducing the satisfaction of its liveness guarantees, as in the work of Abadi and Lamport \cite{Abadi95toplas}.
This approach is suitable for verification, because the model is available at that stage.
However, in the automated construction of specifications for synthesis, we prefer to quantify over time, instead of describing immediate behavior.
Therefore, we desire to be able to reason about dependencies of liveness properties between modules, with minimal reliance on the implementation, i.e., on safety properties.
Stark \cite{Stark86} proposed a proof rule for assume-guarantee reasoning about a non-circular set of liveness properties.
McMillan \cite{McMillan99charme} introduced a proof rule for circular reasoning about liveness.
However, this proof system is intended for verification, so it relies on the availability of a model.
It requires the definition of a proof lattice, and introduces graph edges that {\em consume} time, as a means to break simultaneity cycles.
The method we propose in this work constructs specifications that can have dependencies of liveness goals, but in a way that avoids circularity.
It is discussed in \cref{sec:proposed-approach}.

The assumption-guarantee paradigm has since evolved, and renamed several times.
Meyer \cite{Meyer92c} called the paradigm {\em design by contract}, and supported its use for abstracting software libraries, and validating the correct operation of software.
The notion of a contract generalizes assume-guarantee reasoning, because a contract can have several forms.
For example, it may come in the form of an {\em interface automaton} \cite{Alfaro01}, which offers only an implicit description of assumptions, as those environments that can be successfully connected to the interface.
The interface automaton abstracts the internal details of a module, and serves as its “surface appearance” towards other modules.

More recently, contracts have been proposed for specifying the design of systems with both physical and computational aspects \cite{Benveniste12}.
In this context, contracts are used broadly, as an umbrella term that encompasses both interface theories and assume-guarantee contracts \cite{Nuzzo14memocode,Benveniste12}, with extensions to timed and probabilistic specifications.
A proof system for verifying that a set of contracts refines a contract for the composite system has been proposed in \cite{Cimatti12euromicro}.
A verification tool of contract refinement using an SMT solver is described in \cite{Cimatti13ase}.
This body of work focuses mainly on using, or manipulating, existing contracts.
We are interested in {\em constructing} contracts.

\section{Games}
\label{sec:literature-review}

In this section, we review relevant results from the literature on games of infinite duration.
The literature is extensive, so we restrict to a sample that we consider representative.
The problem of constructing a module that exhibits a desired set of behaviors in time can be solved with algorithms that solve games.
There are different types of games, depending on:
\begin{itemize}
	\item how many transducers are being constructed inside a single system,
	\item the order of player choice,
	\item the winning condition,
	\item the visibility of variables, and
	\item the number of players.
\end{itemize}
Games can be turn-based, where a single player moves in each time step, or concurrent \cite{Alur02jacm,Alfaro00}.
In synchronous games, turns are taken with a fixed schedule, whereas asynchronous games are scheduled dynamically by a dedicated player called scheduler \cite{Finkbeiner07lopstr}.

If we want to construct a single transducer, then the synthesis problem is {\em centralized}.
Synchronous centralized synthesis from LTL has time complexity doubly exponential in the length of the formula \cite{Pnueli89popl}, and polynomial in the number of states.
By restricting to a less expressive fragment of LTL, the complexity can be lowered to polynomial in the formula \cite{Piterman06}.
Asynchronous centralized synthesis does not yield to such a reduction \cite{Pnueli09memocode}.
Partial information games pose a challenge similar to full LTL properties, due to the need for a powerset-like construction \cite{Kupferman00atl}.
To avoid this route, alternative methods have been developed \cite{Kupferman05focs}, that use universal co-B\"{u}chi automata, instead of determinization, and antichains \cite{Wulf06hscc}.

If we want to construct several communicating transducers to obtain some collective behavior, then synthesis is called {\em distributed}.
Of major importance in distributed synthesis is who talks to whom, and how much, called the communication architecture.
A distributed game with full information is in essence a centralized synthesis problem.
Distributed {\em synchronous} games with partial information are undecidable \cite{Pnueli90focs}, unless we restrict the communication architecture to avoid information forks \cite{Finkbeiner05lics}, or restrict the specifications to limited fragments of LTL \cite{Chatterjee13fmcad}.
{\em Bounded synthesis} circumnavigates this intractability by searching for systems with a priori bounded memory \cite{Finkbeiner13sttt}.
Asynchronous distributed synthesis is undecidable \cite{Finkbeiner07lopstr}.

Besides distributed co-synthesis of fixed transducers, the more general notion of {\em assume-guarantee} synthesis \cite{Chatterjee07tacas} constructs transducers that can interface with a complete set of other transducers, as described by an assumption property.
This is the same viewpoint with the approach proposed here.
A difference is that we are interested to synthesize temporal properties with quantification (liveness), instead of directly transducers.
Besides, note that “distributed” in the literature means constructing multiple transducers.
In contrast, we are interested in “distributed” also in the sense that the modules will be synthesized {\em separately}.
Thus, in the problem we consider, distributed synthesis with full information does not reduce to centralized synthesis.

Another body of work relevant to our effort is the construction of assumptions that make an unrealizable problem realizable.
The methods originally developed for this purpose have been targeted at compositional verification, and use the $L^\star$ algorithm for learning deterministic automata \cite{Cobleigh03tacas}, and implemented also symbolically \cite{Nam06atva}.
Later work addressed synthesis, with the theory for a solution proposed in \cite{Chatterjee08concur}, on which our work builds.
This approach separates the construction of assumptions into safety and liveness.
The safety assumption is obtained by property {\em closure}, which also plays a key role in the composition theorem presented in \cite{Abadi95toplas}.

Methods that use opponent strategies to refine the assumptions of a generalized Streett(1) specification, searching over syntactic patterns were proposed in \cite{Li11memocode,Alur13fmcad}.
The syntactic approach of \cite{Alur13fmcad} was used in \cite{Alur15tacas} to refine assume-guarantee specifications of coupled modules.
However, that work cannot handle circularly connected modules, thus neither circular liveness dependencies.
Other approaches aim at identifying the root causes of unrealizability in demanding guarantees \cite{Konighofer13sttt}.
A comprehensive survey can be found in \cite{Bloem14synt}.

\section{Proposed approach}
\label{sec:proposed-approach}

This report proposes a method for constructing assume-guarantee specifications for a set of modules.
The resulting specifications must be {\em realizable} \cite{Abadi1989realizable}, i.e., for each module, there should exist a transducer that implements its specification.
The required behavior of each module is described by a contract over a set of variables that can change values in time.
We choose linear temporal logic (LTL) \cite{Pnueli77} to describe contract specifications.
The specification of a module includes a partition of variables into inputs (uncontrolled by the module), and outputs (controlled by the module), as well as the primitive goals that the module must achieve, but no assumptions yet.
These goals form an overall objective that the resulting contracts should satisfy.
At this stage, the goals may be insufficient to ensure cooperation of the modules with each other.
In other words, the specifier defines guarantees for each module, and the proposed method introduces assumptions that ensure realizability.
Note that each property introduced as an assumption in the contract of some module, will also become a guarantee in the contract of some another module.

We assume that, if we were to construct a single transducer that controls the variables of all modules, then such a transducer exists.
This requires that the conjoined goals be satisfiable.
If the goals are unsatisfiable, then the algorithm diagnoses so, but cannot resolve the conflicts.
Such a resolution would be arbitrary, because it alters the design intent that a human defined, so it should be performed by a human.

As noted in \cref{sec:literature-review}, synthesis from LTL specifications is intractable.
For this reason, we restrict our effort to an LTL fragment that is less expressive, but still practically useful, while allowing synthesis in time polynomial in the number of states, and in the size of the specification formulae.
The selected fragment is known as generalized reactivity of rank 1, GR(1) \cite{Piterman06}, and describes generalized Streett games with one pair, comprised of a persistence and an acceptance property.
This restriction aims at making efficient the synthesis phase, after the contracts have been constructed, as well as the construction of the contracts themselves.
It is a trade-off between expressive power and complexity.
It corresponds to considering the bottom level in the Borel hierarchy of sets of behaviors, as sequences \cite{Manna90podc}.

We model a composite system as a game with multiple players, each representing a module.
In \cref{sec:property-closure}, the winning set is computed for the case of a centralized transducer, also known as the {\em cooperative} winning set.
This is used as a safety assumption for all modules, in order to prevent any module from forcing the system to exit the set from where another module has a winning strategy.

For each module, and each recurrence goal, the winning set in the game with that goal is computed in \cref{sec:nested-games}.
If the winning region is smaller than the cooperative winning set, then weak liveness assumptions are computed for the other players, until reaching a new fixpoint.
These assumptions must be {\em unconditionally realizable}, to prevent trivial realizability of a particular game.

The predicates in the resulting contract are represented symbolically, as binary decision diagrams (BDDs).
This is in contrast with syntactic approaches for constructing assumptions.
Syntactic approaches are restricted to the subset of specifications producible by the chosen grammar template, thus are incomplete.
In contrast, our semantic approach always obtains a solution, if one exists.
The trade-off is that the resulting properties do not have a syntactic form digestible by humans.
The semantic contracts that we construct correspond to a view of contracts as an intermediate result, to be consumed by synthesis algorithms that will construct each individual module, potentially after a refinement of the contract by addition of local, internal, requirements.

As discussed in \cref{sec:design-by-contract}, a challenge in modular specification is reasoning about liveness.
An assume-guarantee contract is intended to remain as declarative as possible.
However, there are behaviors that, if specified declaratively, lead to cyclic dependency of liveness properties.
For this reason, we structure the constructed specifications in a way that avoids circular dependencies of liveness requirements.
This requires imposing a sequencing order on the liveness properties involved.
In verification, the implementation itself is used as reference for enforcing this sequencing.
In temporal logic, it is possible to achieve this purpose by explicitly introducing auxiliary variables.
We avoid introducing such additional variables, because they increase the state space and can be regarded as a limited form of synthesis.
Instead, we alter the specification structure, from flat to nested.
For each liveness goal, nesting is introduced in the form of a stack of games.
Each game in the stack has a reachability objective, and separate assumptions.
Winning one of these games leads higher in the stack, until the top is reached.
The top game can be won directly, and leads to the recurrence goal.
The reliance on safety is in that each subgame is defined on a subset of the states.
In this way, the composite system is prevented from regressing backwards, to a previous game, and progress towards the recurrence goal is ensured.

	\section{Preliminaries}

\subsection{Turn-based synchronous games}
\label{sec:turn-based-sync-games}

We consider turn-based synchronous games with two players \cite{Alur02jacm,Alfaro00}.
The results can be extended to multiple players.
We do not consider concurrent games, because they are not determined, and a strategy can require an infinite amount of memory \cite{Alfaro01}.

The situation in a game is represented by a number of variables.
An assignment to these variables is called the {\em state} of the game.
The game evolves by a sequence of state changes.
If, in each state change, only a single player changes its own state, then the game is called {\em turn-based} \cite{Alur02jacm}.
It is synchronous if the players take turns in a fixed order.

In a game with two players, we will refer to the two players by the indices 0 and 1.
In some cases, we will also use the notation of indexing the players with the letters $e$ (environment) and $s$ (system), instead of numbers.
This is more readable when we discuss operations that consider one player as the system of interest, and lumps the remaining players as the environment of that player.

The state comprises of variables in the set $\mathcal{V}$.
Each player can read all the variables, i.e., it has full information.
Each player can write only those variables that she owns, with the exception of variable $i$.
Player $j$ owns the set $\mathcal{X}_j$ of variables.
In addition, each player increments the auxiliary variable $i$, used to track turns.
So, $\mathcal{V} = \{i\} \cup \bigcup_{i = 0}^{n - 1} \mathcal{X}_i$.

By $x_i$ we will denote both the tuple of symbols in $\mathcal{X}_i$, as well as a tuple of values assigned to those symbols.
In its own turns, player $i$ chooses a next assignment $x_i' \in \mathcal{X}_i'$.
The set of all states $2^{\mathcal{V} }$ is denoted by $S$.
For a set of variable symbols $X$, define the set of assignments $\dbr{\top}_X \triangleq 2^X$.
A predicate $f$ indicates a set
\myeq{
	\dbr{f}_X
	\triangleq
	\{
		u \in 2^X
			\vert\;
		f(u)
	\}.
}

The game can be represented by a {\em game graph}, with nodes partitioned between the two players.
At each node, only one of the two players moves.
The game graph is bipartite, because the game is turn-based.
Note that bipartiteness is necessary\footnote{
	Any game graph can be converted to a bipartite one, by introducing intermediate nodes.}
later, for switching between players when constructing a nested game.
The player that moves first can be selected later, after computing the winning sets, when constructing the transducers.

Each node in the game graph is represented by a tuple $(x, i)$, where:
\begin{itemize}
	\item $x_i \in 2^{\mathcal{X}_i}$ is an assignment for the variables owned by player $i$, and the aggregate state $x = (x_0, x_1, \dots, x_{n - 1})$.
	\item $i \in I \triangleq \naturals_{< n}$ is an index that signifies the player that takes a turn from $(x, i)$.
\end{itemize}
The transition relation of player $i$ is $\hat{\rho}_i(x, x_i')$, where $\hat{\rho}_i$ is an action formula (a Boolean formula over primed and unprimed variables) \cite{Lamport94toplas}.
Player $i$ moves from the node $(x, i)$, by assigning values to variables in $x_i$.
Let $\bar{x}_j$ denote (either a tuple of, or an assignment to) variables in $\bigcup_{ i = 0, i \neq j }^{ n - 1 } \mathcal{X}_i$
We will try to minimize use of the term “state”, because it can be confusing.

\begin{remark}
A (synchronous) interleaving representation \cite{Abadi95toplas} is used here for the game, because it is symmetric and emphasizes the turn-based semantics.
As observed in \cite{Abadi95toplas}, an interleaving representation can be easier to reason about.
In the literature about GR(1) games, typically a non-interleaving representation is used.
In a non-interleaving representation, the combination of primed and unprimed variables captures whose turn it is to play (the role served by the integer variable $i$).
In that representation, player 0 moves from a valutation of $(x_0, x_1)$, and player 1 moves from $(x_0', x_1)$.
Note that the scheduling variable $i$ is shared-write by all players.
\end{remark}

\subsection{Integrals}

In this section, we consider preimage functions induced by the transition relations $\rho_i$.
These functions result from different quantification of the variables.
Depending on the source and target set, several variants can be defined.
We will refer to predicates and the sets they represent interchangeably.

\begin{definition}[Predecessors]
Given a predicate $F$ over $\mathcal{V}$, the {\em existential predecessors} of $F$ are those nodes, from where the set $\dbr{F}_\mathcal{V}$ can be reached with one transition in the game graph,
\myeq{
\Pre_j(F)
\triangleq
\lambda x.\;
\lambda i.\;
(i = j) \wedge
\exists x_j'.\;
\hat{\rho}_j(x, x_j') \wedge
F\vert_{x_j' / x_j}(\bar{x}_j, x_j', j \oplus_n 1),
}
where $j \oplus_n k \triangleq (j + k) \mod{n}$.
Denote $\Pre(F) \triangleq \bigvee_{j \in I} \Pre_j(F)$ the predecessors resulting from moves by all players.
\end{definition}

The semantics of the least fixpoint operator $\mu X.\; f(X)$ is defined as
\myeq{
	\dbr{X_k}_M^\mathcal{E}
	&\triangleq
	\begin{cases}
		\emptyset,
		\quad k = 0
			\\
		\dbr{ f(X) }_M^{ \mathcal{E}[X \leftarrow \dbr{X_{k - 1} }_M] },
		\quad k > 0
	\end{cases}
		\quad
	\dbr{ \mu X.\; f(X) }_M^\mathcal{E}
	&\triangleq
	\bigcup_{k = 0}^\infty X_k,
}
where $M$ a set of variables, and $\mathcal{E}: \{X, \dots\} \rightarrow \dbr{\top}_M$ is an assignment that keeps track of the fixpoint iteration.
The notation $\mathcal{E}[X \leftarrow \dbr{h}_M]$ denotes the modification of $\mathcal{E}$ to assign the set $\dbr{h}_M$ to variable $X$.

\begin{definition}[Iterated predecessors]
The iterated predecessor relation yields the nodes that can reach the set $\dbr{F}$ under {\em some} behavior of the players, or are already in the set $\dbr{F}$, i.e.,
\myeq{
	\Pre^\ast(F)
	\triangleq
	\mu X.\;
	F \vee \Pre(X).
}
\end{definition}
Note that the set $\dbr{\Pre^\ast(F) }$ contains the nodes from where the players can {\em cooperate} to reach the set $\dbr{F}$.
Where clear from the context, we will call both $\Pre$ and $\Pre^\ast$ “predecessor” sets.

\begin{definition}[Controllable predecessors]
The {\em controllable predecessors} of $F$ for player $j$ are those nodes from where player $j$ can force a visit to the set $\dbr{F}_\mathcal{V}$ in the next logic time step, irrespective of how the other players move, i.e.,
\myeq{
	\CPre_j(F)
	\triangleq
	\lambda x.\;
	\lambda i.\;
	\neg^{i \neq j} \exists x_i'.\;
	\hat{\rho}_i(x, x_i') \wedge \neg^{i \neq j} F\vert_{ x_i' / x_i }(\bar{x}_i, x_i', i \oplus_n 1).
}
\end{definition}
For example, for player $j = 0$, it is
\myeq{
	\CPre_0(F)
	= \lambda x.\; \lambda i.\;
		&\big(
			(i = 0) \wedge \neg^{0 \neq 0} \exists x_0'.\;
			\hat{\rho}_0(x, x_0') \wedge \neg^{0 \neq 0} F\vert_{ x_0' / x_0 }(\bar{x}_0, x_0', 0 \oplus_2 1)
		\big)
			\vee
		\\
		&\big(
			(i = 1) \wedge \neg^{1 \neq 0} \exists x_1'.\;
			\hat{\rho}_1(x, x_1') \wedge \neg^{1 \neq 0} F\vert_{ x_1' / x_1 }(\bar{x}_1, x_1', 1 \oplus_2 1)
		\big)
	\\
	= \lambda x.\; \lambda i.\;
		&\big(
			(i = 0) \wedge \neg^0 \exists x_0'.\;
			\hat{\rho}_0(x, x_0') \wedge \neg^0 F\vert_{ x_0' / x_0 }(\bar{x}_0, x_0', 1)
		\big)
			\vee
		\\
		&\big(
			(i = 1) \wedge \neg^1 \exists x_1'.\;
			\hat{\rho}_1(x, x_1') \wedge \neg^1 F\vert_{ x_1' / x_1 }(\bar{x}_1, x_1', 0)
		\big)
	\\
	= \lambda x.\; \lambda i.\;
		&\big(
			(i = 0) \wedge \exists x_0'.\;
			\hat{\rho}_0(x, x_0') \wedge F\vert_{ x_0' / x_0 }(\bar{x}_0, x_0', 1)
		\big)
			\vee
		\\
		&\big(
			(i = 1) \wedge \forall x_1'.\;
			\hat{\rho}_1(x, x_1') \rightarrow F\vert_{ x_1' / x_1 }(\bar{x}_1, x_1', 0)
		\big).
}
As defined here, the operator $\CPre$ is the predicate version of that defined in \cite{Thomas08npgi}.
An attractor contains nodes from where player $j$ can force its way to the set $\dbr{F}$.

\begin{definition}[Attractor]
The {\em attractor} $\Attr_j(F)$ for player $j$ is the set of all nodes, from where the system can force a future visit to the set $\dbr{F}$, or is already in $\dbr{F}$,
\myeq{
	\Attr_j(F)
	\triangleq
	\mu X.\;
	F \vee \CPre_j(X).
}
As alternative notation, let $\CPre_j^\ast(F) \triangleq \Attr_j(F)$.
\end{definition}

\subsection{Linear temporal logic}

Linear temporal logic \cite{Pnueli77} with past \cite{Lichtenstein85clp} is an extension of Boolean logic used to reason about temporal modalities over sequences.
The temporal operators:
\begin{itemize}
	\item “next” $\nxt$,
	\item “previous” $\previous$,
	\item “until” $\until$, and
	\item “since” $\since$
\end{itemize}
suffice to define the other operators \cite{Pnueli77,Baier08}.
Let $AP$ be a set of propositional variable symbols, with values in $\booleans\triangleq\{\bot, \top\}$.
A well-formed LTL formula is inductively defined by
\myeq{
	\varphi
	::=
	p
	&\;|\; \neg \varphi
	\;|\; p \wedge p
	\\&
	\;|\; \nxt \varphi
	\;|\; \varphi \until \varphi
	\\&
	\;|\; \previous \varphi
	\;|\; \varphi \since \varphi.
}
It is modeled by a sequence (word) of variable assignments $w: \naturals \rightarrow \booleans^{AP}$.
Here, we define informally the operators that we will use.
The formula $\always p$ holds if $p$ is forever true, $\eventually p$ if $p$ becomes true in some non-past time.
The {\em weak} “previous” formula $\weakprevious p \triangleq \neg \previous \neg p$ is true if a previous time step does not exist, or $p$ is true in the previous time step.
In contrast, $\previous p$ is true if a previous time step does exist, and $p$ is true then.

\subsection{Interleaving representation of a Streett(1) game}
\label{sec:games-in-logic}

We will use an in interleaving representation \cite{Abadi95toplas}, with the notation defined in \cref{sec:turn-based-sync-games}.
In an interleaving representation of a turn-based game, a single player moves in each logic time step.
In a synchronous game, players move in a fixed order.
This order will be enforced by using the auxiliary variable $i$, as index of the player that should move in the current logic time step.

In a game, each player is assigned a property to realize.
A game structure collects the initial conditions, actions, and liveness goals of each player.
Two-player game structures in a non-interleaving representation are defined in \cite{Bloem12jcss}.
The property to be realized by the player of interest is defined there accordingly.

In an interleaving representation, a generalized reactivity(1) property \cite{Piterman06} to be realized by player $j$ can be described as follows.
Define
\myeq{
	\rho_j(x, x_j', i)
	&\triangleq
	\ite\big(
		i \neq j,\;
		x_j' = x_j,\;
		\hat{\rho_i} \wedge (i' = i \oplus_n 1)
	\big)
		\\
	\bar{\rho}_j(x, \bar{x}_j', i)
	&\triangleq
	\bigwedge_{ k \in J \setminus \{j\} }
		\rho_k(x, x_k', i)
}
In a two-player game, it is
\myeq{
	\bar{\rho}_j(x, \bar{x}_j', i)
	=
	\bigwedge_{ k \in \{0, 1\} \setminus \{j\} }
		\rho_k(x, \bar{x}_k', i)
	=
	\rho_{ 1 - j }(x, \bar{x}_{1 - j}', i).
}

\begin{definition}[Generalized reactivity(1)]
\deflab{def:gr1-interleaving-repr}
Assume that, for $i \in I$, each $\rho_i(x, x_i')$ is an action formula, as defined in \cref{sec:turn-based-sync-games}.
Let $j \in I$ be the index of a player.
Assume that, for $k \in I_P \subset \naturals$, each $P_{j, k}(x, i)$ is an assertion (a Boolean formula over unprimed variables), and similarly for $R_{j, r}(x, i)$.
Then, the LTL formula
\myeq{
	\varphi_{G, j}
	\triangleq
	\begin{matrix*}[l]
	&\wedge\;
	\always
		\big(
			(\weakprevious \historically \bar{\rho}_j)
				\rightarrow
			\rho_j
		\big)
		\\
	&\wedge\;
		\begin{matrix*}[l]
			\big(
				\always \bar{\rho}_j \wedge
				\bigwedge_k
					\always \eventually \neg P_{j, k}
			\big)
			\rightarrow
				\bigwedge_r
				\always \eventually R_{j, r}.
		\end{matrix*}
	\end{matrix*}
}
describes a GR(1) property for player $j$.
\end{definition}

For symmetry, the initial conditions have been omitted above.
Initial conditions require selecting the player that moves first, and their consideration can be delayed until the phase of constructing a winning strategy.
Observe that the action $\rho_i$ can depend on the variables $x, x_i'$, but is independent of the variables $\bar{x}_i'$.

As a shorthand for the above, we define strict implication between two temporal logic formulae in a (synchronous) interleaving representation of a game.
\begin{definition}[Strict implication]
\deflab{def:strict-implication}
Let $\rho_e, \rho_s, P_k, R_r$ be actions (or assertions).
Define the {\em strict implication} operator $\strictarrow$ as
\myeq{
	\big(
	\underbrace{
		\always \rho_e
			\wedge
		\bigwedge_k
			\always \eventually \neg P_k
	}_{assumption}
	\big)
		\strictarrow
	\big(
	\underbrace{
		\always \rho_s
			\wedge
		\bigwedge_r
			\always \eventually R_r
	}_{guarantee}
	\big)
	\triangleq
	\begin{matrix*}[l]
		&\wedge\; \always
			\big(
				(\weakprevious \historically \rho_e)
					\rightarrow
				\rho_s
			\big)
				\\
		&\wedge\;
				\big(
					\always \rho_e \wedge
					\bigwedge_k
						\always \eventually \neg P_k
				\big)
					\rightarrow
				\bigwedge_r
					\always \eventually R_r.
	\end{matrix*}
}
\end{definition}
The antecedent constrains the other players, and the consequent the player under consideration.
For a non-interleaving representation, Strict implication was defined in \cite{Bloem12jcss}.
Unless the action-fairness pairs are machine closed, and the actions are complete, the strict implication operator $\strictarrow$ differs from the \textsc{TLA} while-plus operator $\tlawhilep$ \cite{Abadi95toplas,Lamport02}.

With \cref{def:strict-implication}, we can rewrite \cref{def:gr1-interleaving-repr} using strict implication
\myeq{
	\varphi_{G, j}
	=
	\begin{matrix*}[l]
		&\wedge\;
			\always \bar{\rho}_j
			\\
		&\wedge\;
			\bigwedge_k
				\always \eventually \neg P_{j, k}
	\end{matrix*}\;
		\strictarrow
	\begin{matrix*}[l]
		&\wedge\;
			\always \rho_j
			\\
		&\wedge\;
			\bigwedge_r
				\always \eventually R_{j, r}.
	\end{matrix*}
}

\section{Property closure}
\label{sec:property-closure}

\subsection{Cooperative winning set}
\label{sec:cooperative-winning-set}

In the following, we will present an algorithm for computing pairs of specifications that allow players to cooperate.
For that purpose, some definitions are needed.
Let $\Sigma$ be a suitable alphabet, for example, $\Sigma = 2^\mathcal{V}$.
The set $\Sigma^\star$ denotes finite sequences of elements in $\Sigma$.
The set $\Sigma^\omega$ denotes infinite sequences of elements in $\Sigma$.
The elements of a sequence are indexed by integers, starting at 0.
For a sequence $w \in \Sigma^\omega$, the subsequence that starts at element $i$ and ends at element $j$ (inclusive) is denoted by $w[i \dots j]$.

\begin{definition}[\cite{Abadi91tcs,Abadi95toplas}]
A {\em behavior} or {\em property} $P \subseteq \Sigma^\omega$ is a set of infinite sequences.
\end{definition}

\begin{definition}[Prefix set \cite{Chatterjee08concur}]
The {\em prefix set} of a property $P$ is defined as
\myeq{
	\mathrm{Pref}(P)
	\triangleq
	\{
		\sigma \in \Sigma^\ast
			\vert\;
		\exists w \in P.\;
		\sigma = w[0 \dots \abs{\sigma} - 1]
	\}.
}
\end{definition}

\begin{definition}[Limit set \cite{Chatterjee08concur}]
Given a property $P \subseteq \Sigma^\star$, the set of limits of property $P$ in property $Q$ is defined as
\myeq{
	\mathrm{Safety}_Q(P)
	\triangleq
	\left\{
		w \in Q
			\vert\;
		\forall k \in \naturals.\;
		w[0 \dots k] \in P
	\right\}.
}
If the subscript $Q$ is omitted, then $Q = \Sigma^\omega$, i.e.,
\myeq{
	\mathrm{Safety}(P)
	\triangleq
	\mathrm{Safety}_{\Sigma^\omega}(P).
}
\end{definition}

\begin{definition}[Relative closure \cite{Naylor82,Abadi95toplas}]
\deflab{def:relative-closure}
The {\em closure} of a property $P \subseteq \Sigma^\omega$ with respect to another property $Q \subseteq \Sigma^\omega$ is defined as
\myeq{
	\mathcal{C}_Q(P)
	\triangleq
	\mathrm{Safety}_Q(\mathrm{Pref}(P) )
	=
	\left\{
		w \in Q
			\vert\;
		\forall k \in \naturals.\;
		\exists \sigma \in P.\;
		w[0 \dots k] = \sigma[0 \dots k]
	\right\}
}
If the subscript $Q$ is omitted, then $Q = \Sigma^\omega$, i.e.,
\myeq{
	\mathcal{C}(P)
	\triangleq \mathcal{C}_{\Sigma^\omega}(P).
}
For brevity, define $\overline{P} \triangleq \mathcal{C}(P)$.
\end{definition}
The definition $\mathcal{C}_{\Sigma^\omega}(P)$ corresponds to $\mathcal{C}(P)$ in \cite{Abadi95toplas}.
The closure of a property is with respect to the topology induced by the metric that measures similarity by the length of the longest common prefix between two sequences.

\begin{definition}
Assume that $P \subseteq \Sigma^\omega \cup \Sigma^\star$ is a property.
Define the set of letters that appear in any word in property $P$ as
\myeq{
	\mathrm{States}(P)
	\triangleq
	\left\{
		s \in \Sigma
			\vert\;
		\exists w \in P.\;
		\exists k \in \naturals.\;
		s = w[k]
	\right\}.
}
\end{definition}
The definition of closure implies that $\mathrm{States}(P) = \mathrm{States}(\cl{P})$.

\begin{definition}[\cite{Thomas95stacs,Chatterjee08concur}]
Assume that $F \subseteq \Sigma$ is a set of letters, and $\hat{\rho}_j, j \in I$ a collection of actions (transition relations).
Then, the {\em safe words} are those in the set
\myeq{
	\mathrm{Safe}(F)
	\triangleq
	\Big\{
		w \in \Sigma^\omega
			\vert\;
		\forall k \in \naturals.\;
		w[k] \in F \wedge
			\bigwedge_{j \in I}
			\Bigg(
				(w[k]\vert_i = j)
					\rightarrow
				\begin{array}{l}
					\wedge\; w[k + 1]\vert_i = j \oplus_n 1
						\\
					\wedge\; w[k]\vert_{\bar{x}_j} = w[k + 1]\vert_{\bar{x}_j}
						\\
					\wedge\; \hat{\rho}_j(w[k], w[k + 1])
				\end{array}
			\Bigg)
	\Big\}.
}
\end{definition}

The map $\mathrm{States}$ projects a sequence on the state space.
In the opposite direction, the map $\mathrm{Safe}$ yields the largest invariant subset of a given safe set, under the transition relations.

\begin{definition}[\cite{Chatterjee08concur}]
The {\em cooperative winning set} is the set of nodes in the game graph, from where the players can cooperate to satisfy their objectives.
In a turn-based synchronous game with $n$ players, with objectives $\varphi_j, j \in I$ (that include the transition relations $\rho_j$), it is
\myeq{
	\mathrm{Coop}\Big(\bigwedge_{j \in I} \varphi_j\Big)
	\triangleq
	\Big\{
		u \in \Sigma = 2^\mathcal{V}
			\vert\;
		\exists w \in
		\mathcal{L}\Big(
			\bigwedge_{j \in I} \varphi_j
		\Big).\;
		w[0] = u
	\Big\}.
}
\end{definition}
In other words, the cooperative winning set is the set of nodes from where a centralized controller has a winning strategy.
If the objectives $\varphi_j$ do not include initial conditions\footnote{
	When computing the winning set in a game graph, initial conditions are neglected.
	They are accounted for later, during construction of a transducer.}
(i.e., are tail-closed), then
\myeq{
	\mathrm{Coop}\Big(\bigwedge_{j \in I} \varphi_j\Big)
	= \mathrm{States}\Big( \bigcap_{j \in I} \mathcal{L}(\varphi_j) \Big).
}
The closure of the conjoined specifications is equal to the safe words defined by the cooperative winning set.
This follows from
\myeq{
	\mathrm{States}\big(
		\bigcap_{j \in I}
			\mathcal{L}(\varphi_j)
	\big)
	=
	\mathrm{States}\big(
		\cl{
			\bigcap_{j \in I}
				\mathcal{L}(\varphi_j)
		}
	\big),
}
which implies that
\myeq{
	\mathrm{Coop}\Big(
		\bigwedge_{j \in I} \varphi_j
	\Big)
	=
	\mathrm{States}\Big(
		\cl{
			\bigcap_{j \in I}
				\mathcal{L}(\varphi_j)
		}
	\Big)
		\implies
	\mathrm{Safe}\Big(
		\mathrm{Coop}\Big(
			\bigwedge_{j \in I} \varphi_j
		\Big)
	\Big)
	=
	\mathrm{Safe}\Big(
		\mathrm{States}\Big(
			\cl{
				\bigcap_{j \in I}
					\mathcal{L}(\varphi_j)
			}
		\Big)
	\Big).
}
Observing that each $\varphi_j$ includes $\always \rho_j$, it follows that
$\mathrm{Safe}\Big(
	\mathrm{States}\big(
		\cl{
			\bigcap_{j \in I}
				\mathcal{L}(\varphi_j)
		}
	\big)
\Big)
=
\cl{
	\bigcap_{j \in I}
		\mathcal{L}(\varphi_j)
}$,
therefore
\myeq{
	\mathrm{Safe}\Big(
		\mathrm{Coop}\Big(
			\bigwedge_{j \in I} \varphi_j
		\Big)
	\Big)
	=
	\cl{
		\bigcap_{j \in I}
			\mathcal{L}(\varphi_j)
	}.
}

Define the recurrence formulae
$\mathrm{WF}_j
\triangleq
\bigwedge_r \always \eventually G_{j, r}$, for $j \in I$.
For each player $j$, assume that it has as objective property described by the formula
\myeq{
\label{eq:primitive-specs}
	\varphi_j
	\triangleq \always \rho_j \wedge \mathrm{WF}_j.
}
The property $\varphi$ is in the GR(1) fragment of LTL, so it defines a generalized Streett game of rank 1 (unconditional, i.e., w/o assumptions).
The objectives $\varphi_j$ may be unrealizable.
For each objective $\varphi_j$, we are interested in constructing assumptions that make it realizable.
These assumptions will become objectives for the other agents.
Note that, at this stage
there are no persistence objectives (i.e., no recurrence assumptions yet).

The cooperative winning set can be computed by the fixpoint formula
\myeq{
\label{eq:cooperative-winning-set-fixpoint}
	\mathrm{Coop}\Big(\bigwedge_{j \in I} \varphi_j\Big)
	=
	\nu
	\begin{bmatrix}
		Z_0\\
		Z_1\\
		\vdots\\
		Z_N
	\end{bmatrix}
	.\;
	\begin{bmatrix}
		\Pre^\ast\big( G_{0, 0} \wedge \Pre(Z_1) \big)
			\\
		\Pre^\ast\big( G_{0, 1} \wedge \Pre(Z_2) \big)
			\\
		\vdots
			\\
		\Pre^\ast\big( G_{n - 1, N_{n - 1} - 1 } \wedge \Pre(Z_0) \big)
	\end{bmatrix}
	=
	\nu Z.\;
		\bigwedge_{ j = 0 }^{ n - 1 }
		\bigwedge_{ r = 0 }^{ N_j - 1 }
		\Pre^\ast\big(
			G_{j, r} \wedge \Pre(Z)
		\big).
}
The above computation of the fixpoint involves the recurrence goals of all players.
The aim of decomposing a large system is to modularize the design effort.
This motivates parallelizing the above fixpoint computation.

A slightly different arrangement is also possible.
The goals of each player can be grouped into a vectorized subformula, as follows
\myeq{
	\nu Z.\;
		\bigwedge_{ j = 0 }^{ n - 1 }
		\nu Z_j.\;
		Z \wedge
		\bigwedge_{ r = 0 }^{ N_j - 1 }
		\Pre^\ast\big(
			G_{j, r} \wedge \Pre(Z_j)
		\big)
}
This is expected to increase the sharing of subformulae, because of the overlap of support sets among objectives of a single player.
It is motivated, in part, by the observations of \cref{sec:computing-closure}.
In \cref{sec:computing-closure}, it is shown that the outer fixpoint will be delayed from converging only by states that are live for each objective separately, but not for all objectives jointly.
By increasing coupling between goals, the rate of convergence improves, while still parallelizing the computation, with a granularity at the level of players, instead of individual recurrence goals.
Regarding the variable order, postponing the interaction of BDDs for iterates associated with goals of different players is expected to reduce the coupling between variables, and thus reduce the cost and improve the effectiveness of BDD variable reordering.

\subsection{Computing the closure}
\label{sec:computing-closure}

In this section, we prove that
\myeq{
\label{eq:coop-equivalent-mu-calculus-formula}
	\mathrm{Coop}\Big(\bigwedge_{j \in I} \varphi_j\Big)
	=
	\nu Z.\;
		\bigwedge_{ j = 0 }^{ n - 1 }
		\nu Z_j.\;
		Z \wedge
		\bigwedge_{ r = 0 }^{ N_j - 1 }
		\Pre^\ast\big(
			G_{j, r} \wedge \Pre(Z_j)
		\big).
}
This equality is a consequence of results about vectorized $\mu$-calculus \cite{Lichtensteon91phd}.
Nonetheless, a direct proof is presented below, that gives a better picture of how the sets change during the iteration.

From \cref{sec:cooperative-winning-set}, recall that
$
	\mathrm{Safe}\Big(
		\mathrm{Coop}\Big(
			\bigwedge_{j \in I} \varphi_j
		\Big)
	\Big)
	=
	\cl{
		\bigcap_{j \in I}
			\mathcal{L}(\varphi_j)
	}.
$
In other words, given a conjunction of properties $\mathcal{L}(\varphi_0 \wedge \varphi_1 \wedge \cdots \varphi_{n - 1})$, its closure $\cl{\mathcal{L}(\varphi_0 \wedge \varphi_1 \wedge \cdots \varphi_{n - 1})}$ is equal to the infinite words generated by the restriction of the transition relation to the cooperative winning set.
For this reason, we refer to the closure $\cl{\bigcap_{j \in I} \mathcal{L}(\varphi_j) }$ and the cooperative winning set $\mathrm{Coop}\Big(\bigwedge_{j \in I} \varphi_j\Big)$ interchangeably.

From \cref{eq:cooperative-winning-set-fixpoint}, it suffices to prove that
\myeq{
	\nu Z.\;
		\bigwedge_{ j = 0 }^{ n - 1 }
		\nu Z_j.\;
		Z \wedge
		\bigwedge_{ r = 0 }^{ N_j - 1 }
		\Pre^\ast\big(
			G_{j, r} \wedge \Pre(Z_j)
		\big)
	=
	\nu Z.\;
		\bigwedge_{ j = 0 }^{ n - 1 }
		\bigwedge_{ r = 0 }^{ N_j - 1 }
		\Pre^\ast\big(
			G_{j, r} \wedge \Pre(Z)
		\big).
}
This is equivalent to proving that $\cl{\bigcap_{j \in I} \mathcal{L}(\varphi_j) }$ is equal to the fixpoint iteration that alternates between taking closure and intersection.

\begin{proposition}
\prolab{prop:closure-of-intersections-subset}
For the properties defined by the formulae $\{\varphi_i\}_{i < n}$, the closure of the intersection is a subset of the intersection of closures, i.e.,
$
	\cl{\bigcap_{i=0}^{n - 1} \mathcal{L}(\varphi_i)}
	\subseteq
	\bigcap_{i=0}^{n - 1} \cl{\mathcal{L}(\varphi_i)}.
$
\end{proposition}

\begin{figure}
\small
	\centering
	\includesvg[width=0.5\textwidth]{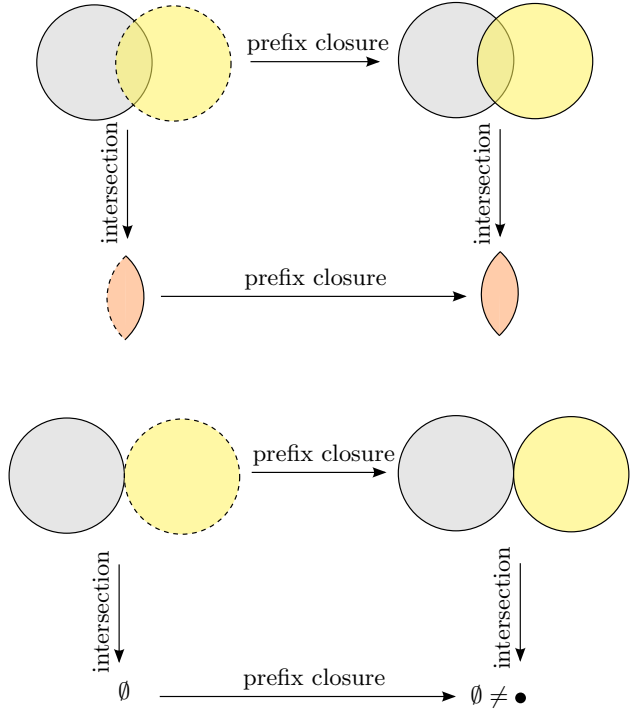}
	\caption{In general, the closure of intersection differs from the intersection of closures.}
	\label{fig:closure_of_intersection}
\end{figure}

The obstruction in parallelizing the computation is that, in general, the opposite containment does {\em not} hold. In that case, the difference arises due to words on the boundary of some property, as proved by the following.

\begin{proposition}
\prolab{prop:excluded-boundary}
Assume that the closure of intersection differs from the intersection of closures, i.e.,
$
	\cl{\bigcap_{i=0}^{n - 1} \lang{\varphi_i}}
	\neq
	\bigcap_{i=0}^{n - 1} \cl{\lang{\varphi_i}}.
$
Then, for each word
$
	w
	\left(\bigcap_{i=0}^{n - 1} \cl{\lang{\varphi_i}}\right)
	\setminus
	\cl{\bigcap_{i=0}^{n - 1} \lang{\varphi_i}},
$
there exists some property $\lang{\varphi_j}$, such that $w$ is on the excluded boundary of property $\lang{\varphi_j}$, i.e.,
$
	w
	\in
	\partial\lang{\varphi_j} \setminus \lang{\varphi_j}.
$
\end{proposition}
In any ball around a word $w$ in the boundary $\partial\lang{\varphi_j}$, there exists some word $z$ in the property $\lang{\varphi_j}$.
It follows that, for any prefix $p$ of word $w$, there exists some word $z \in \lang{\varphi_j}$ that has the prefix $p$.
As a result, the word $w$ is safe with respect to $\lang{\varphi_j}$, but not live.

Next, we define the iteration that corresponds to \cref{eq:coop-equivalent-mu-calculus-formula}, and prove that it converges to the cooperative winning set.

\begin{definition}
\deflab{def:iterate-intersection-closure}
Define $P_j \triangleq \mathcal{L}(\varphi_j)$.
Initialize $Q^0 \triangleq \Sigma^\omega$, and iterate for $k \in \naturals$
\myeq{
	R_j^k
	&\triangleq
	\cl{ Q^k \cap P_j },
		\\
	Q^{k + 1}
	&\triangleq
	\bigcap_{j \in I}
	R_j^k.
}
\end{definition}

We are interested in proving that the iteration of \cref{def:iterate-intersection-closure} reaches as fixpoint the set $\cl{\bigcap_{j \in I} P_j }$.
For this purpose, we will prove that
\begin{itemize}
	\item $\cl{\bigcap_{j \in I} P_j } \subseteq Q^k$ remains invariant (\cref{prop:invariant}), and
	\item if the current iterate $Q^k$ differs from $\cl{\bigcap_{j \in I} P_j }$, then $\abs{\mathrm{States}(Q^{k + 1} ) } < \abs{\mathrm{States}(Q^k ) }$ (\cref{prop:variant-closure-intersection}).
\end{itemize}

\begin{proposition}[Invariant]
\prolab{prop:invariant}
For all $k \in \naturals$, $\cl{\bigcap_{j \in I} P_j} \subseteq Q^k$.
\end{proposition}
\begin{proof}
By induction:

\paragraph{Case $k = 0$}
It is
$
	\cl{\bigcap_{j \in I} P_j}
	\subseteq
	\Sigma^\omega
	= Q_0
$.

\paragraph{Case $k > 0$}
Assume that $\cl{\bigcap_{j \in I} P_j} \subseteq Q^k$.
We will prove that $\cl{\bigcap_{j \in I} P_j} \subseteq Q^{k + 1}$.
By definition of the iterates
\myeq{
\label{eq:reduce-to-ambient-closure}
	Q^{k + 1}
	= \bigcap_{j \in I} R_j^k
	= \bigcap_{j \in I} \cl{Q^k \cap P_j}.
}
By the induction hypothesis,
\myeq{
\label{eq:aux-containment}
	\cl{\bigcap_{i \in I} P_i}
	\subseteq
	Q^k
		\implies
	P_j \cap \cl{\bigcap_{i \in I} P_i}
	\subseteq
	Q^k \cap P_j
		\implies
	\cl{P_j \cap \cl{\bigcap_{i \in I} P_i} }
	\subseteq
	\cl{Q^k \cap P_j}.
}
Therefore, it suffices to prove that $\cl{\bigcap_{i \in I} P_i} \subseteq \cl{P_j \cap \cl{\bigcap_{i \in I} P_i} }$.
It is
\myeq{
	\bigcap_{i \in I} P_i
	\subseteq
	P_j
		\implies
	P_j \cap \bigcap_{i \in I} P_i
	=
	\bigcap_{i \in I} P_i,
}
so
\myeq{
	P_j \cap \cl{\bigcap_{i \in I} P_i}
		&=
	P_j
	\cap
	\Big(
		(\bigcap_{i \in I} P_i)
		\cup
		(\partial \bigcap_{i \in I} P_i)
	\Big)
		=
	\Big(P_j \cap \bigcap_{i \in I} P_i)\Big)
	\cup
	\Big(P_j \cap \partial \bigcap_{i \in I} P_i\Big)
		\\
	&=\bigcap_{i \in I} P_i
	\cup
	\Big(P_j \cap \partial \bigcap_{i \in I} P_i\Big)
		\implies
		\\
	\cl{P_j \cap \cl{\bigcap_{i \in I} P_i} }
	&=
	\cl{\bigcap_{i \in I} P_i \cup
		\big(P_j \cap \partial \bigcap_{i \in I} P_i}\big)
	=
	\cl{\bigcap_{i \in I} P_i} \cup \cl{P_j \cap \partial \bigcap_{i \in I} P_i}
		\implies
		\\
	\cl{\bigcap_{i \in I} P_i}
	&\subseteq
	\cl{P_j \cap \cl{\bigcap_{i \in I} P_i} }
}
By the above result, \cref{eq:aux-containment}, and the induction hypothesis $\cl{\bigcap_{j \in I} P_j} \subseteq Q^k$, it follows that
\myeq{
	\forall j \in I.\;
	\cl{\bigcap_{i \in I} P_i}
	\subseteq
	\cl{P_j \cap \cl{\bigcap_{i \in I} P_i} }
	\subseteq
	\cl{Q^k \cap P_j}
		\implies
	\cl{\bigcap_{i \in I} P_i}
	\subseteq
	\bigcap_{j \in I} \cl{Q^k \cap P_j}
}
Using \cref{eq:reduce-to-ambient-closure}, it follows that
\myeq{
	\cl{\bigcap_{i \in I} P_i}
	\subseteq
	\bigcap_{j \in I} \cl{Q^k \cap P_j}
	=
	Q^{k + 1}.
}
This is the inductive claim.
\end{proof}

\begin{proposition}[Variant]
\prolab{prop:variant-closure-intersection}
If $\cl{\bigcap_{j \in I} P_j} \neq Q^k$, then $\abs{\mathrm{States}(Q^{k + 1} ) } < \abs{\mathrm{States}(Q^k ) }$.
\end{proposition}
\begin{proof}
By definition of the iterates, $Q^{k + 1} = \bigcap_{j \in I} \cl{Q^k \cap P_j}$.
The closure $\cl{Q^0} = \cl{\Sigma^\omega} = \Sigma^\omega = Q^0$, so the set $Q^0$ is closed.
As the intersection of closed sets, the set $Q^k, k > 0$ is closed.
It is
\myeq{
Q^k \cap P_j
\subseteq
Q^k
	\implies
\cl{Q^k \cap P_j}
\subseteq
\cl{Q^k}
=
Q^k
	\implies
Q^{k + 1}
=
\bigcap_{j \in I}
\cl{Q^k \cap P_j}
\subseteq
Q^k.
}

It remains to prove that $Q^{k + 1} \neq Q^k$.
We will show that taking the closures $\cl{Q^k \cap P_j}$ will yield at least one set $\mathrm{States}(R_j^k) \subsetneq \mathrm{States}(Q^k)$.
By \cref{prop:invariant}, $\cl{\bigcap_{j \in I} P_j} \subseteq Q^k$, and by hypothesis they are not equal.
So, the difference $K \triangleq Q^k \setminus \cl{\bigcap_{j \in I} P_j}$ is non-empty.
By induction, the containment $Q^{k + 1} \subseteq Q^k$ implies that, for any $k > 0$, it is $Q^k \subseteq Q^1 = \bigcap_{j \in I} \cl{P_j}$.
So, $K \subseteq \bigcap_{j \in I} \cl{P_j}$.
This result is analogous to \cref{prop:excluded-boundary}, but for an arbitrary iteration along the computation.

Consider any word $w \in K$.
By the previous, $w \in \bigcap_{j \in I} \cl{P_j}$, and $w \notin \cl{\bigcap_{j \in I} P_j}$.
The game graph is finite, so, by the pigeonhole principle, the word $w$ has a finite prefix and a finite cycle as suffix.
Denote by $M$ the non-empty set of nodes in the suffix.
The word $w \in \bigcap_{j \in I} \cl{P_j}$, so, from each node in $M$, for each $j \in I$, a strongly connected component (SCC) that intersects all recurrence sets of $P_j$ is reachable.
The word $w$ is not in $\cl{\bigcap_{j \in I} P_j}$.
So an SCC that intersects the recurrence sets of all properties is {\em not} reachable from any node in $M$.

Define $S_j$ the SCC that intersects the recurrence sets of $P_j$ and is reachable from $M$.
The set $R_j^k = \cl{Q^k \cap P_j}$, so nodes in $\mathrm{States}(R_j^k)$ can reach only the intersection $S_j \cap \mathrm{States}(Q^k)$.
If $S_j \cap \mathrm{States}(Q^k) = \emptyset$, and there are no other SCCs that intersect a $P_j$ and are reachable from $M$, then the nodes in $M \subseteq \mathrm{States}(Q^k)$ are not in $\mathrm{States}(R_j^k)$, and the claim holds.

Suppose that $S_j \cap \mathrm{States}(Q^k) \neq \emptyset$.
Consider the nodes in $S_j \cap \mathrm{States}(Q^k)$.
These nodes are in $\mathrm{States}(Q^k)$.
So the same arguments apply, as those we developed for nodes in $M$.
This leads to new SCCs, that form a directed acyclic graph (DAG).
By finiteness of the game graph, the induction will terminate.

Consider a leaf of the DAG.
It is an SCC terminal in $\mathrm{States}(Q^k)$, that does not intersect at least one recurrence set, of at least one property $P_j$.
(If not, then the SCC would satisfy $\bigcap_{j \in I} P_j$.
By construction, the leaf SCC is reachable from the nodes in $M$ (suffix).
This implies that from nodes in $M$, an SCC satisfying $\bigcap_{j \in I} P_j$ is reachable.
It follows that $w$ is in $\cl{\bigcap_{j \in I} P_j}$.
This contradicts the definition of $w$, as a word not in $\cl{\bigcap_{j \in I} P_j}$.)
Therefore, there is at least one recurrence set of $P_j$, which is unreachable from the nodes in the leaf SCC, without exiting the set $\mathrm{States}(Q^k)$.
It follows that none of these nodes is contained in $\mathrm{States}(\cl{Q^k \cap P_j} )$.
These nodes are in $\mathrm{States}(Q^k)$, so $R_j^k = \cl{Q^k \cap P_j} \subsetneq Q^k$.
\end{proof}

We have proved the following.
\begin{theorem}
The closure of intersection $\cl{\bigcap_{j \in I} \mathcal{L}(\varphi_j) }$ is equal to the fixpoint of the iterated intersection of closures $Q^{k + 1} = \bigcap_{j \in I} \cl{ Q^k \cap \mathcal{L}(\varphi_j) }$, starting from $Q^0 = \Sigma^\omega$.
\end{theorem}

After the cooperative winning set $C = \mathrm{Coop}(\bigwedge_{j \in I} \varphi_j )$ has been computed, each transition relation $\rho_j$ is restricted to it, by conjoining it with $\rho_C \triangleq C \wedge C'$.
As proved in \cite{Chatterjee08concur} for the case of two players, the restriction to the cooperative winning set satisfies two properties:
\begin{enumerate}
	\item it is not restrictive, because it removes edges from the transition relation $\rho_i$ of player $i$, only if they lead outside the closure with respect to some other player $\rho_j$.
	\item among all non-restrictive properties, the restriction to the cooperative winning set is minimal, as measured by the cardinality of the edges removed from the game graph.
\end{enumerate}
In addition, the safety property $\always C$ is added to the assumptions of each agent.
The specifications become (redefining \cref{eq:primitive-specs} by adding a safety assumption)
\myeq{
	\varphi_j
	\triangleq
	(\always \rho_C)
		\strictarrow
	\big(
		\always \rho_j \wedge
		\always \rho_C \wedge
		\mathrm{WF}_j
	\big).
}

	\section{Construction of weak fairness assumptions for a single goal}

In this section, we introduce the main elements for the proposed algorithm, for the case of two agents.
Let us consider a single recurrence goal.
More than one goals are treated by constructing a transducer that cycles through them, and communicating to other players the currently pursued goal.
This is described at the end of \cref{sec:nested-games}.
This need for coordination of pursued goals is unavoidable, because, otherwise, livelock arises naturally.

Our objective is to find assumptions that allow covering the cooperative winning set.
This problem has been solved for a single agent, and full LTL, in \cite{Chatterjee08concur}.
Here, we are interested in assumptions restricted to the GR(1) fragment, and in multiple players.
Recall that in \cref{sec:property-closure} we conjoined the transition relations with the requirement that each player stays inside the cooperatively winning set $C$, similarly to \cite{Chatterjee08concur}.

Let $G = \always \eventually G_{j = 0, r = 0}$ be the recurrence goal of interest, of player 0.
Player 0 can force a visit to the set $G$ from any node in the attractor $A_0 \triangleq \Attr_0(G)$.
But $A_0$ may not cover the cooperative winning set $C$.
By the definition of $C$, the set $A_0$ is reachable from $C \setminus A_0$.
Since nodes in $C \setminus A_0$ do not belong to $A_0$, player 0 cannot force a transition from $C \setminus A_0$ to $A_0$.
By determinacy of turn-based synchronous games with full information, player 1 must be able to force such a transition.
It follows that the attractor $\Attr_1(A_0)$ is non-empty.
This form of argument is reminiscent of the solution of parity games \cite{McNaughton93apal}.

We want to construct an {\em unconditional} assumption that player 0 makes about player 1.
Unconditional means that player 1 should be able to realize the assumption, without assuming any liveness property about player 0.
If it needed to assume a liveness property about player 0, that would create circularity, causing trivial realizability.

A first attempt could be $\always \eventually (\Attr_1(A_0) \rightarrow A_0)$.
This is insufficient, because player 0 may be able to exit the set $\Attr_1(A_0)$, but go to $\neg A_0$ -- not to $A_0$.
So player 0 must be able to restrict player 1 inside a subset $K \subseteq \Attr_1(A_0)$, until player 0 forces its way to $A_0$, obliged by an assumption of the form $\always \eventually (K \rightarrow A_0)$.
The inclusion $K \subseteq \Attr_1(A_0)$ ensures that player 1 cannot trap player 0 inside $K$, which would cause trivial realizability.
Such an assumption may not exist, a case that is addressed later.

This exist requirement can be formalized by defining\footnote{
	The greatest fixpoint operator $\nu$ is defined as $\nu X.\; f(X) \triangleq \neg \mu X.\; \neg f(X)$.}
the {\em controlled-escape} subset of a set $S$,
\myeq{
	\Trap_j(S, E)
	\triangleq
	\nu X.\;
	E \vee (\CPre_j(X) \wedge S).
}
The set $\Trap_j(S, E)$ contains those nodes, from where player $j$ can force to either remain inside $\Trap_j(S, E)$, or move to $E$, or is already in $E$.
Note that $\Trap_j(S \vee E, \bot)$ is different, because it requires the ability to remain inside $S \vee E$.

Define $B_0 \triangleq \Attr_1(A_0)$, and $r_0 \triangleq (\Trap_0(B_0, A_0) \wedge B_0) \setminus A_0$.
With this definition of a trap, we can now define the assumption of player 0 about player 1
\myeq{
	\always \eventually (A_0 \vee \neg r_0)
	=
	\always \eventually (r_0 \rightarrow A_0).
}
This assumption extends the winning set of player 0, only if $\dbr{r_0} \neq \emptyset$.
Otherwise, the assumption is not useful, and we need to either:
\begin{enumerate}
	\item introduce a safety assumption that refers to additional variables, or
	\item define the specification as a nested game.
\end{enumerate}
In the following, we elaborate on these claims.

\subsection{The role of machine closure}

In \cref{sec:property-closure}, we conjoined the transition relations with a safety requirement to remain inside the cooperative winning set $C$.
In this section, we give an example, demonstrating that absence of closure can lead to a contract unrealizable by player 1, together with a contract that is trivially realizable by player 0.

In \cref{fig:example_lack_of_closure}, nodes from where player 0 (player 1) moves are denoted by disks (boxes).
Player 0 wants $\always \eventually G_{0, 0}$, and player 1 $\always \eventually G_{1, 0}$.
The goal $G_{1, 0}$ is not reachable from nodes $c, d$, so these nodes are not in the cooperative winning set $C$.
Suppose that we ignored this, and used the transition relation $\rho_1$, as given by the specifier.
Then, player 0 would think that player 1 can continue from node $b$ to node $c$, towards $d$.
In other words, player 0 will compute a larger attractor $\Attr_1(A_0)$ for player 1.
Taking into account the closure of the goal $\always \eventually G_{1, 0}$ by restricting $\rho_1$ to $\tilde{\rho}_1$, player 1 cannot take the transition $(b, c)$.

So, the property $\always \eventually ((a \vee b \vee e \vee f) \rightarrow A_0)$, assumed by player 1, is not realizable by player 0.
If player 0 knows about the goal $G_{1, 0}$ of player 1, then the game with this assumption becomes trivially realizable by player 0, from the nodes $a, b, e$.
Otherwise, the unrealizable contract will result in player 0 possibly choosing always the transition $(a, b)$ in vain, awaiting that player 0 will take $(b, c)$.
In both cases, the design fails.

To avoid trivial realizability (that corresponds to circularity of liveness assumptions), we need to introduce a nested game, where player 1 assumes that player 1 will eventually transition to $f$.
In this particular game, the nested game would have been avoided, had we conjoined with $\rho_C$, in order to ensure closure.
This demonstrates that lack of closure can manifest itself as superfluous liveness assumptions that, due to possible circularity, give rise to unnecessary game nesting (nesting will be defined later).
The pair $(\always (\rho_0 \wedge \rho_1), \always \eventually G_{1, 0})$ is not {\em machine closed} \cite{Abadi95toplas}, because $\cl{\always (\rho_0 \wedge \rho_1) \wedge \always \eventually G_{1, 0} } \neq \always (\rho_0 \wedge \rho_1)$, i.e., the property $\always \eventually G_{1, 0}$ introduces a safety constrain on $\always \rho_1$.

This superfluous nesting of games can result also due to variable hiding.
If some variables of player 1 are hidden from player 0, then it may be the case that player 1 can traverse $(b, c)$ only when its internal state allows so.

\begin{figure}
	\centering
	\includesvg[width=0.7\textwidth]{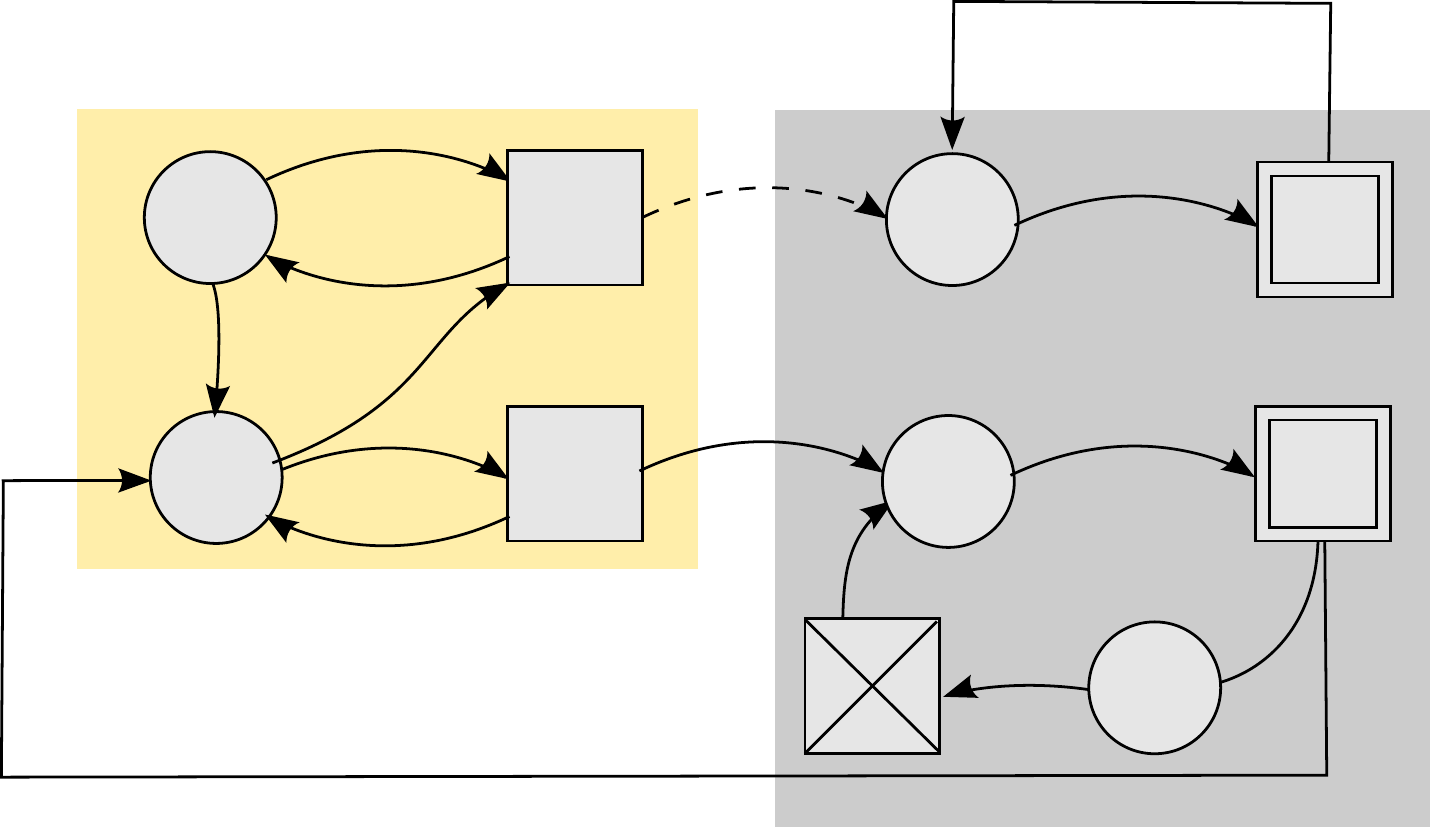}
	\caption{Example that demonstrates that lack of closure information can lead to the need for introducing additional nested games.}
	\label{fig:example_lack_of_closure}
\end{figure}

\subsection{Nonexistence of weak fairness assumptions over nodes}

Suppose that $\dbr{\Trap_0(B_0, A_0) } = \emptyset$.
This means that player 0 cannot keep player 1 in any subset of the attractor $\Attr_1(A_0)$.
We will use two counterexamples, to prove that, if we restrict the assumptions to recurrence properties in the GR(1) fragment, then it is {\em impossible} to cover the cooperative winning region.
Recall that in GR(1), a recurrence property includes a predicate over nodes, but not edges.

\begin{proposition}
\prolab{prop:transition-relations-hold}
\end{proposition}
\assume{An infinite sequence $w \not\models \always \rho_0 \wedge \always \rho_1$.}
\prove{For at least one of the two players, for any property $\varphi$ of the form of \cref{def:gr1-interleaving-repr}, the sequence $w$ does not model $\varphi$.}
\begin{pf2proof}
\begin{proof}
\step{<1>1}{
	There exists a $k \in \naturals$ such that $w[k \dots k + 1] \not\models \rho_0 \wedge \rho_1$.}
\step{<1>2}{
	\pick\ the minimal $k \in \naturals$ such that $w[k \dots k + 1] \not\models \rho_0 \wedge \rho_1$. }
	\begin{proof}
	\pf\ By \stepref{<1>1}, the set of $k$ with this property is non-empty, countable, and bounded from below.
		So a minimal $k$ exists.
	\end{proof}
\step{<1>3}{
	\case{$w[k \dots k + 1] \not\models \rho_1$} }
	\begin{proof}
	\step{<2>1}{
		$\forall r \in 0 \dots k - 1.\; w[r \dots r + 1] \models \rho_0 \wedge \rho_1$}
		\begin{proof}
		\pf\ By \stepref{<1>2}, $k$ is the minimal non-negative integer with this property.
		\end{proof}
	\step{<2>2}{
		$w, k \models \weakprevious \historically \rho_0$}
		\begin{proof}
		\pf\ By \stepref{<2>1}.
		\end{proof}
	\step{<2>3}{
		$w, k \not\models \rho_1$}
		\begin{proof}
		\pf\ By \stepref{<1>3}.
		\end{proof}
	\qedstep
		\begin{proof}
		\pf\ By \stepref{<2>2} and \stepref{<2>3},
			\myeq[*]{
				w, k \not\models (\weakprevious \historically \rho_0) \rightarrow \rho_1
					\implies
				w \not\models \always\big(
					(\weakprevious \historically \rho_0)
						\rightarrow
					\rho_1 
				\big).
			}
		\end{proof}
	\end{proof}
\step{<1>4}{
	\case{$w[k \dots k + 1] \not\models \rho_0$} }
	\begin{proof}
	\pf\ Similar to \stepref{<1>3}.
	\end{proof}
\qedstep
	\begin{proof}
	\pf\ By \stepref{<1>2}, the cases \stepref{<1>3} and \stepref{<1>4} are exhaustive.
	\end{proof}
\end{proof}
\end{pf2proof}

\begin{proposition}
\prolab{prop:counterexample-weak-fairness}
\end{proposition}
\assume{
	Define the transition relations $\rho_0, \rho_1$ by the game graph of \cref{fig:counterexample_for_weak_fairness}.
	Define the set of nodes $V \triangleq \{s_0, \dots, s_7\}$.
	Define the goal $G = \{s_6\}$ of player 0.
}
\prove{
	There does not exist a set $\dbr{P} \subseteq V$, such that:
	\begin{enumerate}
		\item the property
		\myeq{
			\varphi_1
			\triangleq
			(\always \rho_1)
				\strictarrow
			\big(
				\always \rho_0 \wedge
				\always \eventually P
			\big),
		}
		be realizable by player 1, and
		\item the property
		\myeq{
			\varphi_0
			\triangleq
			\big(
				\always \rho_0 \wedge
				\always \eventually P
			\big)
				\strictarrow
			\big(
				\always \rho_1 \wedge
				\always \eventually G
			\big)
		}
		be realizable by player 0.
	\end{enumerate}
}

\begin{figure}[t]
	\centering
	\includesvg[width=0.8\textwidth]{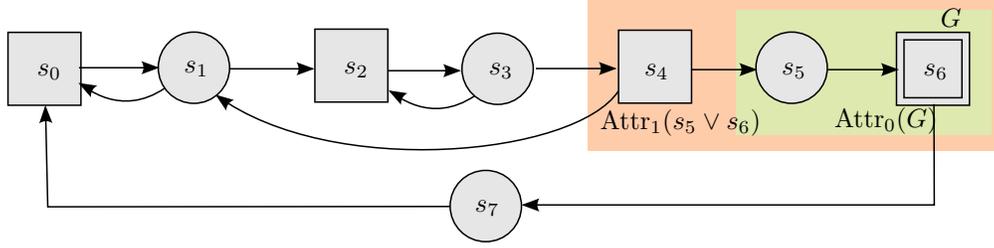}
	\caption{
		There does not exist a weak fairness assumption that suffices for realizability in this example.
		Player 0 (player 1) controls the play at disks (boxes).}
	\label{fig:counterexample_for_weak_fairness}
\end{figure}

\begin{pf2proof}
\begin{proof}
\step{<1>0}{
	$\always (\rho_0 \wedge \rho_1)$.}
	\begin{proof}
	\pf\ By \cref{prop:transition-relations-hold}, if $\always (\rho_0 \wedge \rho_1)$ is false for a play, then $\varphi_0$ or $\varphi_1$ is false.
	\end{proof}
\step{<1>1}{ \case{$\dbr{P} = \emptyset$} }
	\begin{proof}
		\pf\ $\always \eventually P = \always \eventually \bot$ is not realizable by player 1.
	\end{proof}
\step{<1>2}{ \case{$\dbr{P} \neq \emptyset$} }
	\begin{proof}
	\step{<2>1}{ \case{$\dbr{P} \cap \{s_0, s_1\} = \emptyset$} }
		\begin{proof}
		\step{<3>1}{
			$\dbr{P} \cap (V \setminus \{s_0, s_1\}) \neq \emptyset$}
			\begin{proof}
			\pf\ By \stepref{<1>2} and \stepref{<2>1}.
			\end{proof}
		\step{<3>2}{
			$\dbr{P} \cap \{s_2, \dots, s_7\} \neq \emptyset$}
			\begin{proof}
			\pf\ By \stepref{<3>1} and definition of node set $V$.
			\end{proof}
		\step{<3>3}{
			$\always \eventually P$ not realizable by player 1.}
			\begin{proof}
			\step{<4>1}{
				\define{
					Player 0 strategy
					\myeq[*]{
						f
						\triangleq
						\big(s_3 \rightarrow (s_3 \wedge s_4')\big) \wedge
						\big(s_1 \rightarrow (s_1 \wedge s_0')\big)
					}
				}
			}
			\step{<4>2}{
				If player 0 uses the strategy $f$ of \stepref{<4>1}, then all plays violate $\always \eventually P$.}
				\begin{proof}
				\step{<5>1}{
					From the nodes $s_2, \dots, s_7$, the play goes to node $s_1$.}
				\step{<5>2}{
					From node $s_1$, the play is $s_1 (s_0 s_1)^\omega$.}
				\qedstep
					\begin{proof}
					\pf\ By \stepref{<5>2}, any play reaches, and then remains forever in, the set $\{s_0, s_1\}$.
					By \stepref{<2>1}, this play does not intersect $\dbr{P}$, so the play does not satisfy $\always \eventually P$.
					\end{proof}
				\end{proof}
			\qedstep
				\begin{proof}
				\pf\ By \stepref{<4>2}, player 1 cannot realize $\always \eventually P$.
				\end{proof}
			\end{proof}
		\qedstep
			\begin{proof}
			\pf\ By \stepref{<3>3}, the consequent of $\varphi_1$ is false.
			\end{proof}
		\end{proof}
	\step{<2>2}{
		\case{$\dbr{P} \cap \{s_0, s_1\} \neq \emptyset$} }
		\begin{proof}
		\step{<3>1}{
			\case{$\dbr{P} \cap \{s_2, s_3\} = \emptyset$} }
			\begin{proof}
				\step{<4>1}{
					\define{
						Player 0 strategy
						\myeq[*]{
							f
							\triangleq
							\big(s_1 \rightarrow (s_1 \wedge s_2')\big) \wedge
							\big(s_3 \rightarrow (s_3 \wedge s_2')\big)
						}
					}
				}
				\step{<4>2}{
					If player 0 uses the strategy $f$ of \stepref{<4>1}, then all plays violate $\always \eventually P$.}
					\begin{proof}
						\step{<5>1}{From the nodes $s_0, s_4, \dots, s_7$, the play goes to node $s_1$.}
						\step{<5>2}{From node $s_1$, the play is $s_1 (s_2 s_3)^\omega$.}
						\qedstep
							\begin{proof}
							\pf\ By \stepref{<5>2}, the play reaches, and then remains forever in, the set $\{s_2, s_3\}$.
							By \stepref{<3>1}, this play does not intersect $\dbr{P}$, so the play does not satisfy $\always \eventually P$.
							\end{proof}
					\end{proof}
				\qedstep
					\begin{proof}
					\pf\ By \stepref{<4>2}, $\varphi_1$ is false.
					\end{proof}
			\end{proof}
		\step{<3>2}{
			\case{$\dbr{P} \cap \{s_2, s_3\} \neq \emptyset$} }
			\begin{proof}
				\step{<4>1}{
					\define{
						Player 1 strategy
						$
							f
							\triangleq
							s_4 \rightarrow (s_4 \wedge s_1')
						$.
					}
				}
				\step{<4>2}{
					If player 1 uses strategy $f$ of \stepref{<4>1}, and the play is in the set $\{s_0, \dots, s_4\}$, then the play remains in $\{s_0, \dots, s_4\}$ in the next time step.}
					\begin{proof}
					\pf\ The only edge that exits the set $\{s_0, \dots, s_4\}$, and satisfies both $\rho_0$ and $\rho_1$, is $s_4 \wedge s_5'$.
					This player 1 edge is not in the strategy $f$ of \stepref{<4>1}.
					\end{proof}
				\step{<4>3}{
					If a play starts in the set $\{s_5, s_6, s_7\}$, then it reaches the set $\{s_0, \dots, s_4\}$ in a finite number of steps.}
					\begin{proof}
					\pf\ By the definition of $\rho_0, \rho_1$ and \stepref{<1>0}.
					\end{proof}
				\step{<4>4}{
					If player 1 uses strategy $f$ of \stepref{<4>1}, then any play enters the set $\{s_0, \dots, s_4\}$, and then remains in it.}
					\begin{proof}
					\pf\ By \stepref{<4>2} and \stepref{<4>3}.
					\end{proof}
				\step{<4>5}{
					Any play where player 1 uses the strategy $f$ of \stepref{<4>1} satisfies $\always \eventually P.$}
					\begin{proof}
					\step{<5>2}{
						Any play either reaches, and remains forever in, the set $\{s_2, s_3\}$, or it visits node $s_1$.}
						\begin{proof}
						\step{<6>1}{
							It is possible to remain forever in $\{s_2, s_3\}$.}
						\step{<6>2}{
							If the play exits $\{s_2, s_3\}$, then it visits $s_1$.}
							\begin{proof}
							\pf\ The only edge that exits $\{s_2, s_3\}$ is $s_3 \wedge s_4'$.
							By \stepref{<4>1}, the next edge is $s_4 \wedge s_1'$.
							\end{proof}
						\qedstep
							\begin{proof}
							\pf\ By \stepref{<6>1} and \stepref{<6>2}.
							\end{proof}
						\end{proof}
					\step{<5>3}{
						If the play visits node $s_1$, then it either visits both $s_0$ and $s_1$, or both $s_2$ and $s_3$.}
						\begin{proof}
						\pf\ Each edge outgoing from node $s_1$ leads to either $s_0$ and $s_1$, or to $s_2$ and $s_3$.
						\end{proof}
					\step{<5>4}{
						Any play visits the set $\dbr{P}$ infinitely many times.}
						\begin{proof}
						\pf\ By \stepref{<5>2}, \stepref{<5>3}, the play either visits both $s_2$ and $s_3$ infinitely many times, or it reaches $s_1$ infinitely many times, so also either $s_2$ and $s_3$ infinitely many times, or $s_0$ and $s_1$ infinitely many times.
						By \stepref{<2>2} and \stepref{<3>2}, the play visits the set $\dbr{P}$ infinitely many times.
						\end{proof}
					\qedstep
						\begin{proof}
						\pf\ By \stepref{<5>4}, the play satisfies $\always \eventually P$.
						\end{proof}
					\end{proof}
				\qedstep
					\begin{proof}
					\pf\ By \stepref{<4>4} and \stepref{<4>5}, any play where player 1 uses the strategy $f$ satisfies $\always \eventually P$ and violates $\always \eventually G$.
					So, $\varphi_0$ is not true.
					\end{proof}
			\end{proof}
		\qedstep
			\begin{proof}
			\pf\ By \stepref{<3>1} and \stepref{<3>2}.
			\end{proof}
		\end{proof}
	\qedstep
		\begin{proof}
		\pf\ By \stepref{<2>1} and \stepref{<2>2}.
		\end{proof}
	\end{proof}
\qedstep
	\begin{proof}
	\pf\ By \stepref{<1>1} and \stepref{<1>2}.
	\end{proof}
\end{proof}
\end{pf2proof}

We can make a number of observations.
Firstly, there {\em does} exist a weak fairness assumption {\em outside} of the GR(1) fragment, such, that the game of \cref{fig:counterexample_for_weak_fairness} becomes non-trivially realizable.
This weak fairness property is in an extension of the GR(1) fragment with action predicates in recurrence properties \cite{Raman13phd}.

In particular, we have to tell player 0 that it is unfair to, forever, hide in the set $\{s_2, s_3\}$, i.e., $\eventually \always \neg (s_3 \wedge s_2')$.
If we add this property both as an assumption of player 1, and as a guarantee by player 0, then trivial realizability persists, because this is a liveness property (ignoring, for a moment, that this results in a Rabin(1) game).
Thus, it should not be added as a guarantee for player 0.

But we can “subtract” this property from the assumption of player 0.
Consider the desired assume-guarantee pair for player 1
\myeq{
	\eventually \always \neg (s_3 \wedge s_2')
	\rightarrow
	\always \eventually \big( (s_1 \vee s_2 \vee s_3 \vee s_4) \rightarrow s_5\big)
}
Then, merge the antecedent (persistence) and consequent (recurrence) into a single recurrence property
\myeq{
	\always \eventually \big(
		(s_3 \wedge s_2') \vee
		s_5 \vee
		\neg (s_1 \vee s_2 \vee s_3 \vee s_4)
	\big).
}
This property is realizable by player 1, but {\em not} in the GR(1) fragment, because $(s_3 \wedge s_2')$ is an edge.
It is in an extension of GR(1) with edges in liveness properties \cite{Raman13phd}.

Moreover, the above property can be expressed in GR(1), by shifting the above transition formula one step into the past, as $\always \eventually \previous (\dots)$.
This introduces a {\em history variable}, for remembering the past, and a safety property about this variable's update behavior.
Pnueli observes in \cite{Pnueli85lmcs} the equivalence of auxiliary variables, with allowing the past.
We observe that describing in GR(1) this weak fairness property, which involves a transition relation, introduces a safety property, and increases the number of variables.

In general, a weak fairness assumption over edges (of both players) in the game graph can be computed by finding a trap set that is sufficiently large, to prevent player 1 from satisfying the assumption, by going away from the goal desired by player 0 (e.g., the edge $s_4 \wedge s_1'$ in \cref{fig:counterexample_for_weak_fairness}).
Such a set can lead to trivial realizability.
In order to prevent trivial realizability, edges of player 0 that lead away from the goal can be subtracted from the assumption, as we did above with the edge $s_3 \wedge s_2'$.
These edges can be computed by considering consecutive iterates of a reachability computation in the cooperative winning set.

Here, we decide to use the GR(1) fragment, with recurrence properties over nodes, because recurrence assumptions that refer to edges of player 0 need to include all backward leading edges inside the trap set.
Therefore, this type of assumptions explicitly refers to the transition relation, over a set of nodes.
As a result, it leads to more complex and detailed formulae, which are less amenable to simplification, and are less suitable for an extension to cases with hidden variables.

Note that in a non-interleaving representation, both primed and unprimed variables are required to represent nodes from where player 1 moves.
In more detail, player 1 moves from nodes of the form $(x_0', x_1, i)$.
Even though such a representation involves primed variables, in the game graph, these are still nodes, not edges.
Therefore, in a non-interleaving representation, the propositions have the same semantics, but with different syntax.

A more direct approach is to introduce safety, by requiring that $\always \big(s_4 \rightarrow (s_4 \wedge s_5')\big)$.
This resolves the non-determinism, by fixing a choice (undesirable).
However, such a fixed safety assumption may not exist, as proved by the following.

\begin{figure}[t]
	\centering
	\includesvg[width=0.8\textwidth]{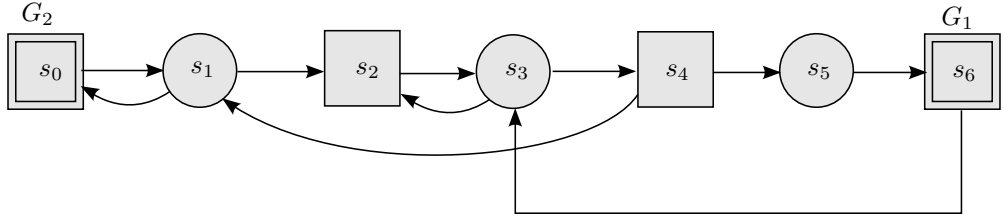}
	\caption{
		There does not exist a realizable GR(1) property that suffices as an assumption in this example, as proved in \cref{lem:no-gr1-assumption}.
		Player 0 (player 1) controls the play at disks (boxes).}
	\label{fig:counterexample_for_safety_assumption}
\end{figure}

\begin{proposition}[Nonexistence of safety]
\prolab{prop:counterexample-for-safety-assumption}
\end{proposition}

\assume{
	Define the transition relations $\rho_0, \rho_1$ by the game graph of \cref{fig:counterexample_for_safety_assumption}.
}
\prove{
	There does not exist a set $\dbr{\rho} \subsetneq \dbr{\rho_1}$, such that player 1 chooses edges that satisfy $\rho$, and
	\myeq{
		\varphi
		\triangleq
		\always(\rho_0 \wedge \rho) \wedge
		\always \eventually G_1 \wedge
		\always \eventually G_2
	}
	is satisfiable (cooperatively by player 0 and player 1).
}

\begin{pf2proof}
\begin{proof}
\step{<1>0}{
	$\always (\rho_0 \wedge \rho)$}
	\begin{proof}
	\pf\ If $\always (\rho_0 \wedge \rho_1)$ is false for a play, then $\varphi$ is false.
	\end{proof}
\step{<1>1}{
	\case{$\rho$ does not include the edge $s_0 \wedge s_1'$.} }
	\begin{proof}
	\step{<2>1}{
		No infinite play visits $\dbr{G_2}$.}
		\begin{proof}
		\pf\ By \stepref{<1>1}, if a play visits node $s_0$, then there is no next node.
		\end{proof}
	\qedstep
		\begin{proof}
		\pf\ By \stepref{<2>1}, no play satisfies the property $\always \eventually G_2$.
		\end{proof}
	\end{proof}
\step{<1>2}{
	\case{$\rho$ does not include the edge $s_2 \wedge s_3'$.} }
	\begin{proof}
	\step{<2>1}{
		If an infinite play visits $\dbr{G_2}$, then it does not satisfy $\always \eventually G_1$.}
		\begin{proof}
		\pf\ By \stepref{<1>2}, no path exists from the set $\dbr{G_2} = \{s_0\}$, to the set $\dbr{G_1} = \{s_6\}$.
		\end{proof}
	\step{<2>2}{
		\case{A play satisfies $\always \eventually G_2$.} }
		\begin{proof}
		\pf\ By \stepref{<2>2}, the play visits $\dbr{G_2}$, so by \stepref{<2>1}, the play violates $\always \eventually G_1$.
		\end{proof}
	\step{<2>3}{
		\case{A play violates $\always \eventually G_2$.} }
		\begin{proof}
		\pf\ By definition of $\varphi$.
		\end{proof}
	\qedstep
		\begin{proof}
		\pf\ By \stepref{<2>2} and \stepref{<2>3}, with such a $\rho$, no play satisfies $\varphi$.
		\end{proof}
	\end{proof}
\step{<1>3}{
	\case{$\rho$ does not include the edge $s_4 \wedge s_1'$.} }
	\begin{proof}
	\step{<2>1}{
		If an infinite play visits $\dbr{G_1}$, then it does not satisfy $\always \eventually G_2$.}
		\begin{proof}
		\pf\ By \stepref{<1>3}, there is no path from $\dbr{G_1}$ to $\dbr{G_2}$.
		\end{proof}
	\step{<2>2}{
		\case{A play satisfies $\always \eventually G_1$.} }
		\begin{proof}
		\pf\ By \stepref{<2>2}, the play visits $\dbr{G_1}$, so by \stepref{<2>1}, the play violates $\always \eventually G_2$.
		\end{proof}
	\step{<2>3}{
		\case{A play violates $\always \eventually G_1$.} }
		\begin{proof}
		\pf\ By definition of $\varphi$.
		\end{proof}
	\qedstep
		\begin{proof}
		\pf\ By \stepref{<2>2} and \stepref{<2>3}, with such a $\rho$, no play satisfies $\varphi$.
		\end{proof}
	\end{proof}
\step{<1>4}{
	\case{$\rho$ does not include the edge $s_4 \wedge s_5'$.} }
	\begin{proof}
	\step{<2>1}{
		If a play visits node $s_6$, then it does not revisit $s_6$.}
		\begin{proof}
		\pf\ By \stepref{<1>4}, there does not exist a path from node $s_6$ to node $s_6$.
		\end{proof}
	\qedstep
		\begin{proof}
		\pf\ By \stepref{<2>1}, no play satisfies $\always \eventually G_1$.
		By definition of $\varphi$, no play satisfies $\varphi$.
		\end{proof}
	\end{proof}
\step{<1>5}{
	\case{$\rho$ does not include the edge $s_6 \wedge s_3'$.} }
	\begin{proof}
	\pf\ Similar to \stepref{<1>1}, but for the set $\dbr{G_1}$.
	\end{proof}
\qedstep
	\begin{proof}
	\step{<2>1}{
		$\rho \subseteq \rho_1$ and $\rho \neq \rho_1$.}
		\begin{proof}
		\pf\ By hypothesis.
		\end{proof}
	\step{<2>2}{
		The transition relation $\rho$ has at least one fewer edge than $\rho_1$.}
		\begin{proof}
		\pf\ By \stepref{<2>1}.
		\end{proof}
	\step{<2>3}{
		The cases \stepref{<1>1}--\stepref{<1>5} are exhaustive.}
		\begin{proof}
		\pf\ By \stepref{<2>2}, the definition, by hypothesis, of $\rho_1$ with 5 edges, and the case statements \stepref{<1>1}--\stepref{<1>5} for those 5 edges.
		\end{proof}
	\qedstep
		\begin{proof}
		\pf\ By \stepref{<2>3}.
		\end{proof}
	\end{proof}
\end{proof}
\end{pf2proof}

In \cref{prop:counterexample-for-safety-assumption}, we proved lack of satisfiability, not lack of mutual realizability.
This condition (for safety here) is stronger than in \cref{prop:counterexample-weak-fairness} (for recurrence there).
The reason is that we will need to combine the result with \cref{prop:nonexistence-recurrence}.
If a property (safety) is not realizable by player 1, then conjoining with another property (recurrence) restricts it further, so it remains unrealizable by player 1.
In general, if a property does not suffice as an assumption for player 0, it is not true that restricting it will yield a property unrealizable by player 0.
However, if a property $P \wedge Q$ is unsatisfiable cooperatively by the two players, then further restriction yields an unsatisfiable property.

Suppose that the property $P$ is realizable by player 1, and the property $P \strictarrow Q$ by player 0.
Use a strategy for each player to control all the variables.
The composite strategy satisfies $P \wedge (P \strictarrow Q)$, so also $P \wedge Q$.
Therefore, the property $P \wedge Q$ is satisfiable cooperatively, a contradiction.
So, the restriction of $P$ to $\hat{P}$ yields assume-guarantee pairs $\hat{P}$ for player 1, and $\hat{P} \strictarrow Q$, of which at least one is not realizable.

We turn now to the nonexistence of a recurrence assumption for the conjoined goals $(\always \eventually G_1) \wedge (\always \eventually G_2)$.
For each of the goals $\always \eventually G_1$ and $\always \eventually G_2$, there exists a recurrence assumption $\always \eventually P$, such that both the formula
$
	\big(
		\always \rho_0
	\big)
		\strictarrow
	\big(
		\always \rho_1 \wedge
		\always \eventually P
	\big)
$
is realizable by player 1, and the formula
$
	\big(
		\always \rho_1 \wedge
		\always \eventually P
	\big)
		\strictarrow
	\big(
		\always \rho_0 \wedge
		\always \eventually G_i
	\big)
$
is realizable by player 2.
In particular,
\begin{itemize}
	\item $\always \eventually s_3$ for $\always \eventually G_2$
	\item $\always \eventually (s_0 \vee s_2)$ for $\always \eventually G_1$.
\end{itemize}
The mutual realizability for these assumptions has been confirmed with a GR(1) synthesizer.

\begin{proposition}[Nonexistence of recurrence]
\prolab{prop:nonexistence-recurrence}
\end{proposition}

\assume{
	Define $\rho_i$ the transition relation of player $i$ by the game of \cref{fig:counterexample_for_safety_assumption}.
	Define the set of nodes $V \triangleq \{s_0, \dots, s_6\}$.
}
\prove{
	For all sets of nodes $P \subseteq V$, for any initial node, either
	\begin{itemize}
		\item the property
		\myeq{
			\varphi_1
			\triangleq
			\big(
				\always \rho_1
			\big)
				\strictarrow
			\big(
				\always \rho_0 \wedge
				\always \eventually P
			\big),
		}
		is not realizable by player 1,
		or
		\item the property
		\myeq{
			\varphi_0
			\triangleq
			\big(
				\always \rho_0 \wedge
				\always \eventually P
			\big)
				\strictarrow
			\big(
				\always \rho_1 \wedge
				\always \eventually G_1 \wedge
				\always \eventually G_2
			\big),
		}
		is not realizable by player 0.
	\end{itemize}
}

\begin{pf2proof}
\begin{proof}
\step{<1>0}{
	$\always (\rho_0 \wedge \rho_1)$}
	\begin{proof}
	\pf\ By \cref{prop:transition-relations-hold}, if $\always (\rho_0 \wedge \rho_1)$ is false for a play, then $\varphi_0$ or $\varphi_1$ is false.
	\end{proof}
\step{<1>1}{
	\case{$\dbr{P} = \emptyset$} }
	\begin{proof}
	\pf\ By \stepref{<1>0} and \stepref{<1>1}, $\varphi_1$ is false.
	\end{proof}
\step{<1>2}{
	\case{$\dbr{P} \neq \emptyset$.} }
	\begin{proof}
	\step{<2>1}{
		\case{$\dbr{P} \cap \{s_2, s_3\} = \emptyset$.} }
		\begin{proof}
		\step{<3>1}{
			\define{Player 0 strategy
				\myeq[*]{
					f
					\triangleq
					\big(s_1 \rightarrow (s_1 \wedge s_2')\big) \wedge
					\big(s_3 \rightarrow (s_3 \wedge s_2')\big).
				}
			}
		}
		\step{<3>2}{
			If player 0 uses strategy $f$, then no play satisfies $\always \eventually P$.}
			\begin{proof}
			\step{<4>1}{
				From the set $\{s_0, s_2, s_4, s_5, s_6\}$, the play visits either node $s_1$, or node $s_3$.}
			\step{<4>2}{
				If player 0 uses strategy $f$, then from the nodes in $\{s_1, s_3\}$, the play is $(s_2 s_3)^\omega$.}
				\begin{proof}
				\pf\ By the strategy $f$ of \stepref{<3>1}.
				\end{proof}
			\step{<4>3}{
				Any play is of the form $s_i^\ast (s_2 s_3)^\omega$.}
				\begin{proof}
				\pf\ By \stepref{<4>1}, \stepref{<4>2}, and that these cases cover $V$.
				\end{proof}
			\qedstep
				\begin{proof}
				\pf\ By \stepref{<4>3} and \stepref{<2>1}.
				\end{proof}
			\end{proof}
		\qedstep
			\begin{proof}
			\pf\ By \stepref{<3>2}, if player 0 uses the strategy $f$ of \stepref{<3>1}, then player 1 cannot realize property $\varphi_1$.
			\end{proof}
		\end{proof}
	\step{<2>2}{
		\case{$\dbr{P} \cap \{s_2, s_3\} \neq \emptyset$.} }
		\begin{proof}
		\step{<3>1}{
			\case{$\dbr{P} \cap \{s_0, s_1\} \neq \emptyset$.} }
			\begin{proof}
			\step{<4>1}{
				\define{Player 1 strategy
					\myeq[*]{
						f
						\triangleq
						s_4 \rightarrow (s_4 \wedge s_1').
					}
				}
			}
			\step{<4>2}{
				If a play visits a node in $\{s_5, s_6\}$, then it later visits the set $\{s_0, \dots, s_4\}$.}
			\step{<4>3}{
				If player 1 uses strategy $f$ and a play visits the set $\{s_0, \dots, s_4\}$, then the play remains forever in it.}
				\begin{proof}
				\step{<5>1}{
					The only edge that exits $\{s_0, \dots, s_4\}$ is $s_4 \wedge s_5'$.}
				\step{<5>2}{
					The edge $s_4 \wedge s_5'$ is not in the strategy $f$ of \stepref{<4>1}.}
				\qedstep
					By \stepref{<5>1} and \stepref{<5>2}.
				\end{proof}
			\step{<4>4}{
				If player 1 uses strategy $f$, then no play visits $G_1$ an infinite number of times.}
				\begin{proof}
				\pf\ By \stepref{<4>2} and \stepref{<4>3}.
				\end{proof}
			\step{<4>5}{
				Any play that remains in $\{s_0, \dots, s_4\}$ visits $s_0$ and $s_1$, or $s_2$ and $s_3$, an infinite number of times.}
			\step{<4>6}{
				Any play that remains in $\{s_0, \dots, s_4\}$ satisfies $\always \eventually P$.}
				\begin{proof}
				\pf\ By \stepref{<4>5}, \stepref{<2>2}, and \stepref{<3>1}.
				\end{proof}
			\step{<4>7}{
				If player 1 uses strategy $f$, then all plays satisfy $\always \eventually P$.}
				\begin{proof}
				\pf\ By \stepref{<4>2}, \stepref{<4>3}, and \stepref{<4>6}.
				\end{proof}
			\qedstep
				\begin{proof}
				\pf By \stepref{<4>4} and \stepref{<4>7}, if player 1 uses strategy $f$, then all plays satisfy $\always \eventually P$ and violate $\always \eventually G_1$.
				By definition of $\varphi_0$, all plays violate $\varphi_0$.
				So, there does not exist a winning strategy for player 0.
				\end{proof}
			\end{proof}
		\step{<3>2}{
			\case{$\dbr{P} \cap \{s_0, s_1\} = \emptyset$.} }
			\begin{proof}
			\step{<4>1}{
				\case{$\dbr{P} \cap \{s_3, \dots, s_6\} = \emptyset$.} }
				\begin{proof}
				\step{<5>1}{
					\define{Player 0 strategy
						\myeq[*]{
							f
							\triangleq
							\big(s_1 \rightarrow (s_1 \wedge s_0')\big) \wedge
							\big(s_3 \rightarrow (s_3 \wedge s_4')\big).
						}
					}
				}
				\step{<5>2}{
					If player 0 uses strategy $f$, then no play visits node $s_2$ infinitely many times.}
					\begin{proof}
					\pf\ If a play starts at node $s_2$, then it leaves $s_2$.
					By \stepref{<5>1}, if a play is not at node $s_2$, then none of the edges incoming to $s_2$ is in the strategy \stepref{<4>1}.
					\end{proof}
				\step{<5>3}{
					$\dbr{P} = \{s_2\}$}
					\begin{proof}
					\pf\ By \stepref{<2>2}, \stepref{<3>2}, and \stepref{<4>1}.
					\end{proof}
				\step{<5>4}{
					If player 0 uses strategy $f$, then no play satisfies $\always \eventually P$.}
					\begin{proof}
					\pf\ By \stepref{<5>2} and \stepref{<5>3}.
					\end{proof}
				\qedstep
					\begin{proof}
					\pf\ By \stepref{<5>4}, if player 0 uses strategy $f$, then all plays violate $\always \eventually P$.
					By definition of $\varphi_1$, all plays violate $\varphi_1$.
					So, there does not exist a winning strategy for player 1.
					\end{proof}
				\end{proof}
			\step{<4>2}{
				\case{$\dbr{P} \cap \{s_3, \dots, s_6\} \neq \emptyset$.} }
				\begin{proof}
				\step{<5>1}{
					\case{Initial node in $\{s_0, s_1\}$.} }
					\begin{proof}
					\step{<6>1}{
						\define{Player 0 strategy
							\myeq[*]{
								f
								\triangleq
								s_1 \rightarrow (s_1 \wedge s_0').
							}
						}
					}
					\step{<6>2}{
						If player 0 uses strategy $f$, then all plays remain in the set $\{s_0, s_1\}$.}
						\begin{proof}
						\pf\ By \stepref{<5>1} and \stepref{<6>1}.
						\end{proof}
					\step{<6>3}{
						If player 0 uses strategy $f$, then all plays violate $\always \eventually P$.}
						\begin{proof}
						\pf\ By \stepref{<6>2} and \stepref{<3>2}.
						\end{proof}
					\qedstep
						\begin{proof}
						\pf\ By definition of $\varphi_1$, and \stepref{<6>3}, if player 0 uses strategy $f$, then all plays violate $\varphi_1$.
						So, there does not exist a winning strategy for player 1.
						\end{proof}
					\end{proof}
				\step{<5>2}{
					\case{Initial node not in $\{s_0, s_1\}$.} }
					\begin{proof}
					\step{<6>1}{
						\define{Player 1 strategy
							\myeq[*]{
								f
								\triangleq
								s_4 \rightarrow (s_4 \wedge s_5').
							}
						}
					}
					\step{<6>2}{
						If player 1 uses strategy $f$, then all plays visit infinitely many times either $s_2$ and $s_3$, or $s_3, s_4, s_5$ and $s_6$.}
						\begin{proof}
						\pf\ By \stepref{<1>0}, \stepref{<5>2}, \stepref{<6>1}, the play remains in the set $\{s_2, \dots, s_6\}$.
						The only cycles in $\{s_2, \dots, s_6\}$ are $s_2, s_3$ and $s_3, s_4, s_5, s_6$.
						By the pigeonhole principle, at least one of these two cycles must be visited an infinite number of times.
						\end{proof}
					\step{<6>3}{
						If player 1 uses strategy $f$, then all plays satisfy $\always \eventually P$.}
						\begin{proof}
						\pf\ By \stepref{<6>2}, \stepref{<2>2} and \stepref{<4>2}.
						\end{proof}
					\step{<6>4}{
						If player 1 uses strategy $f$, then no play visits $\dbr{G_2}$.}
						\begin{proof}
						\step{<7>1}{All plays start outside $\{s_0, s_1\}$.}
							\begin{proof}
							\pf\ By \stepref{<5>2}.
							\end{proof}
						\step{<7>2}{No play that is outside $\{s_0, s_1\}$, enters $\{s_0, s_1\}$.}
							\begin{proof}
							\pf\ By \stepref{<6>1}.
							\end{proof}
						\step{<7>3}{No play visits $\{s_0, s_1\}$.}
							\begin{proof}
							\pf\ By \stepref{<7>1} and \stepref{<7>2}.
							\end{proof}
						\qedstep
							\begin{proof}
							\pf\ By \stepref{<7>3} and the definition of $G_2$.
							\end{proof}
						\end{proof}
					\step{<6>5}{
						If player 1 uses strategy $f$, then no play satisfies $\always \eventually G_2$.}
						\begin{proof}
						\pf\ By \stepref{<6>4}.
						\end{proof}
					\qedstep
						\begin{proof}
						\pf\ By \stepref{<6>3} and \stepref{<6>5}, if player 1 uses strategy $f$, then all plays satisfy $\always \eventually P$ and violate $\always \eventually G_2$.
						By definition of $\varphi_0$, all plays violate $\varphi_0$.
						So, there does not exist a winning strategy for player 0.
						\end{proof}
					\end{proof}
				\qedstep
					\begin{proof}
					\pf\ By \stepref{<5>1} and \stepref{<5>2}, that cover all initial nodes in $V$.
					\end{proof}
				\end{proof}
			\qedstep
				\begin{proof}
				\pf\ By \stepref{<4>1} and \stepref{<4>2}.
				\end{proof}
			\end{proof}
		\qedstep
			\begin{proof}
			\pf\ By \stepref{<3>1} and \stepref{<3>2}.
			\end{proof}
		\end{proof}
	\qedstep
		\begin{proof}
		\pf\ By \stepref{<2>1} and \stepref{<2>2}.
		\end{proof}
	\end{proof}
\qedstep
	\begin{proof}
	\pf\ By \stepref{<1>1} and \stepref{<1>2}.
	\end{proof}
\end{proof}
\end{pf2proof}

\begin{lemma}[Nonexistence of GR(1) assumption]
\lemlab{lem:no-gr1-assumption}
Define the transition relations $\rho_0, \rho_1$ as in the game of \cref{fig:counterexample_for_safety_assumption}.
There does not exist a property $P$ in the GR(1) fragment, such that
\myeq{
	\varphi_1
	\triangleq
	(\always \rho_1)
		\strictarrow
	(\always \rho_0 \wedge P)
}
be realizable by player 1, and
\myeq{
	\varphi_0
	\triangleq
	(
		\always \rho_0 \wedge
		P
	)
		\strictarrow
	(
		\always \rho_1 \wedge
		\always \eventually G_1 \wedge
		\always \eventually G_2
	).
}
be realizable by player 0.
\end{lemma}
\begin{pf2proof}
\begin{proof}
\pf\ By \cref{prop:counterexample-for-safety-assumption} and \cref{prop:nonexistence-recurrence}.
\end{proof}
\end{pf2proof}

This can be avoided, by introducing a goal counter $\mathit{goal}$ as auxiliary variable, and switch between safety assumptions, depending on the counter, e.g., $\always ( (s_4 \wedge \mathit{goal} = 1) \rightarrow (s_4 \wedge s_5'))$.

Here, we decide to not introduce explicitly new variables in the contract, neither safety assumptions that fix choices of edges.
Instead, in \cref{sec:nested-games}, we will define nested games, where the safety assumptions are introduced by partitioning the game graph into sub-games, and avoid explicit reference to extra variables inside the formula.
The purpose served by those extra variables is achieved by structuring the contract into multiple games.

\section{Nested games}
\label{sec:nested-games}

A structured way of isolating conditional assumptions is by partitioning the game into smaller ones.
Each smaller game has its own assumptions, independently of the other games.
This prevents circularity of liveness dependencies.
Each game has one reachability objective: to reach the game that contains it.
Only unconditional liveness assumptions can appear inside each game.
Assumptions that themselves depend on other liveness assumptions become objectives in their own game.
The games partition the game graph.
The approach of nested games is reminiscent of McNaughton's recursive algorithm for solving parity games \cite{McNaughton93apal}.

\begin{figure}
\centering
\includesvg[width=0.6\textwidth]{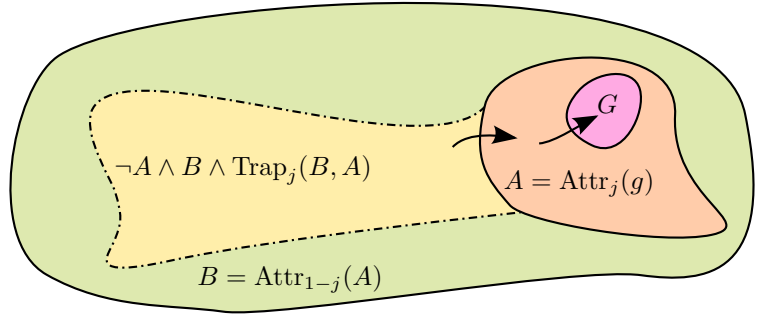}
\caption{The sets (labeled by predicates) computed by \textsc{UnconditionalAssumption} in \cref{algo:nested-game-construction}.}
\label{fig:attr1_trap_attr0}
\end{figure}

\begin{algorithm}[H]
\caption{Construction of nested-game GR(1) specification, for each recurrence goal $G$}
\label{algo:nested-game-construction}
\begin{algorithmic}[1]
\Procedure{GameStack}{$j,\; G,\; \mathit{uncovered},\; \mathit{stack}$}
\State $\mathit{trap} \gets \top$
\State $\mathit{goal} \gets G$
\State $\mathit{stack} \gets \mathrm{set}()$
\While{$\dbr{\mathit{trap} } \neq \emptyset$} \Comment{Create unconditional assumptions, until stuck}
	\State $\mathit{attr}, \mathit{trap} \gets$
		\Call{UnconditionalAssumption}{$j,\; \mathit{goal}$}
	\State $\mathit{goal} \gets \mathit{attr} \cup \mathit{trap}$
	\State $\mathit{assumptions}.\mathrm{add}\big(
		\always \eventually (
			\mathit{trap}
				\rightarrow
			\mathit{attr}
		) \big)$
\EndWhile
\State $\mathit{game} \gets (j,\; \mathit{goal} \wedge \neg G,\; G,\; \mathit{assumptions})$
\State $\mathit{stack}.\mathrm{append}(\mathit{game})$
\State $\mathit{uncovered} \gets \mathit{uncovered} \wedge \neg \mathit{goal}$
\If{$\dbr{\mathit{uncovered} } = \emptyset$}
	\Comment{Covered cooperatively winning set?}
	\State \Return
\EndIf
\State \Call{GameStack}{$1 - j,\; \mathit{goal},\; \mathit{uncovered},\; \mathit{stack}$}
	\Comment{Construct a nested game}
\State \Return
\EndProcedure
\Procedure{UnconditionalAssumption}{$j,\; g$}
\State $A \gets \Attr_j(g)$
\State $B \gets \Attr_{1 - j}(A)$
\State $r \gets \neg A \wedge B \wedge \Trap_j(B, A)$
\State \Return $A,\; r$
\EndProcedure
\end{algorithmic}
\end{algorithm}

\cref{algo:nested-game-construction} computes a stack of nested games, for reaching a goal $G$.
It covers the cooperative winning set $C$, so a later visit to $G$ is always possible, from any node in $C$.
Part of the computation is illustrated in \cref{fig:attr1_trap_attr0}.

\begin{proposition}[\textsc{GameStack} variant]
\prolab{prop:game-stack-variant}
If procedure \textsc{GameStack} calls \textsc{GameStack} (L14), then the set $\dbr{\mathit{uncovered} }$ after L11 in the caller has at least one more node than $\dbr{\mathit{uncovered} }$ after L11 in the callee.
\end{proposition}
\begin{proof}
Consider a call to \textsc{GameStack} by \textsc{GameStack} (L14).
Variables in the caller, and in its last call to \textsc{UnconditionalAssumption} (L6) will be indexed by 1.
Variables in the callee, and in its first call to \textsc{UnconditionalAssumption} (L6) will be indexed by 2.

We will prove that, in the first call of the callee to \textsc{UnconditionalAssumption} (L6), the attractor $A_2 = \Attr_{j_2}(g_2)$ will be strictly larger than $g_2$ (L17).
We need to prove that there is a node outside $g_2$, from where player $j_2$ can force a visit to $g_2$.
It is $g_2 = \mathit{goal}_2$ (L16,6) in the first iteration of the loop (L5).
First iteration implies $\mathit{goal}_2 = G_2$ (L3).
In the caller, $\mathit{goal}_1 = G_2$ (L14,1), so $g_2 = \mathit{goal}_2 = G_2 = \mathit{goal}_1$.
The value of $j_2$ (L17) is $1 - j_1$ in the caller (L16,6,1,14).

In the caller, L14 was reached.
So the loop terminated, implying $\dbr{\mathit{trap}_1} = \emptyset$ (L5).
In the last loop iteration, $\dbr{\mathit{trap}_1} = \emptyset$ implies that $\dbr{\mathit{goal}_1} = \dbr{\mathit{attr}_1}$ (L7).
The return statement (L12) was not executed, so $\dbr{\mathit{uncovered}_{1, L12} } \neq \emptyset$.
By L11, $\dbr{\mathit{uncovered}_{1, L12} } \neq \emptyset$ implies that $\mathit{goal}_1$ does not cover $\mathit{uncovered}_{1, L1}$ at L1.
By $\dbr{\mathit{goal}_1} = \dbr{\mathit{attr}_1}$, it follows that $\dbr{\mathit{attr}_1}$ does not cover $\mathit{uncovered}_{1, L1}$.

It is $A_1 = \mathit{attr}_1$ (L6,20) from the last call to \textsc{UnconditionalAssumption}.
So $A_1 = \Attr_{j_1}(g_1)$ does not cover $\mathit{uncovered}_{1, L1}$.
By definition, $\mathit{uncovered}_{1, L1}$ is a subset of the cooperative winning set, and goal $g_1$ is contained in $A_1$.
So, any node in $\mathit{uncovered}_{1, L1}$ can reach $g_1$, thus also $A_1$.

Suppose that no node of player $(1 - j_1)$ in $\mathit{uncovered}_{1, L12} = \mathit{uncovered}_{1, L1} \wedge \neg \mathit{goal}_1 = \mathit{uncovered}_{1, L1} \wedge \neg A_1$ has an edge that leads to $A_1$.
Then, $\mathit{uncovered}_{1, L1} \wedge \neg A_1$ (non-empty) must contain a node of player $j_1$ that has an edge to $A_1$.
This node must\footnote{
	In a turn-based game, from each node, a single player controls all edges.
	This argument would not hold in a concurrent game, a consequence of lacking determinacy \cite{Alur02jacm,Alfaro01}.}
be in $A_1$, because $A_1$ is an attractor for player $j_1$.
This is a contradiction.
We conclude that at least one node of player $j_2 = 1 - j_1$ is in $\mathit{uncovered}_{1, L1} \wedge \neg A_1$ and has an edge to $A_1$.
This node is outside $A_1 = \mathit{attr}_1 = \mathit{goal}_1 = \mathit{G}_2$, and will be in $A_2 = \Attr_{j_2}(g_2) = \Attr_{1 - j_1}(A_1)$ in the first call to \textsc{UnconditionalAssumption} by the callee.
This proves the claim.
\end{proof}

\begin{proposition}[\textsc{GameStack} Termination]
\prolab{prop:game-stack-invariant}
If the game graph is finite,
then any call to procedure \textsc{GameStack} of \cref{algo:nested-game-construction} terminates.
\end{proposition}
\begin{proof}
A call to \textsc{GameStack} may not terminate for two reasons: the loop or the recursion never terminate.
Suppose that the loop never terminates, so $\dbr{\mathit{trap} } \neq \emptyset$.
It is $\mathit{goal} = g$ (L6,16) and $\dbr{g} \subseteq \dbr{\Attr_j(g) }$ (attractor def) and $A = \Attr_j(g)$ (L17), so
$
\dbr{\mathit{goal} }
\subseteq
\dbr{A}
$.

The set $\dbr{\mathit{trap} } = \dbr{r}$ (L6,20) and $\dbr{r} \cap \dbr{A} = \emptyset$ (L19), so $\dbr{\mathit{trap} } \cap \dbr{A} \neq \emptyset$.
We supposed that $\dbr{\mathit{trap} } \neq \emptyset$, so the set $\dbr{\mathit{trap} }$ contains nodes outside $\dbr{A}$.
By $\dbr{\mathit{goal} } \subseteq \dbr{A}$, it follows that $\dbr{\mathit{trap} }$ contains nodes outside $\dbr{\mathit{goal} }$.

So the set $\dbr{\mathit{goal} }$ increases strictly in each iteration.
By hypothesis, the game graph has a finite number of nodes, so $\mathit{goal}$ will eventually cover the graph, implying that $\dbr{B} = \dbr{A}$ (L18), thus $\dbr{\mathit{trap} } = \dbr{r} \subseteq \dbr{\neg A \wedge B} = \dbr{\neg A \wedge A} = \emptyset$.
This contradicts the supposition $\dbr{\mathit{trap} } = \emptyset$.
So the loop at L5 terminates.

Suppose that the number of recursive calls to \textsc{GameStack} is infinite.
By \cref{prop:game-stack-variant}, with each recursive call to \textsc{GameStack}, the cardinality of the set $\dbr{\mathit{uncovered} }$ decreases by at least one.
We supposed an infinite number of recursive calls, so in some recursive call to \textsc{GameStack}, $\dbr{\mathit{uncovered} } = \emptyset$.
So the guard of L12 becomes true, and that call returns, without any further recursion, a contradiction.
Therefore, the number of recursive calls is finite.
\end{proof}

Upon termination, the algorithm has computed a stack of games, each game is in effect in a subset of the game graph.

The time complexity is at most quadratic in the number of nodes, with time measured by $\CPre_j$ calls.
This complexity follows because of single alternation of least and greatest fixpoints (L17--19).
For each call to \textsc{UnconditionalAssumption} either $\dbr{\textit{trap} } = \emptyset$, so by \cref{prop:game-stack-variant} the next call to \textsc{UnconditionalAssumption} will remove a node from the uncovered ones, or $\dbr{\textit{trap} } \neq \emptyset$ so by \cref{prop:game-stack-invariant}, the current call removes a node from the uncovered ones.
Therefore, \textsc{UnconditionalAssumption} is called at most $2 \abs{\Sigma}$ times.

Each call to \textsc{UnconditionalAssumption} contains two chained attractor computations, and a trap computation.
Each of these can invoke $\CPre_j$ at most $\abs{\Sigma}$ times.
The previous two statements imply that the time complexity is at most quadratic in the number of game graph nodes.

Note that searching for fewer assumptions, inducing a smaller winning set, can be exponentially expensive,
as proved for syntactic recurrence formulae in \cite{Alur15tacas}.
Conceptually, the nesting of games has common elements with modular game graphs \cite{Alur06tcs} and open temporal logic \cite{Banerjee05todaes}.

Let us revisit the example of \cref{fig:counterexample_for_weak_fairness}, to observe the algorithm's execution.
Player 0 wants $\always \eventually G$.
The first call to \textsc{GameStack} will call \textsc{UnconditionalAssumption}.
Player 0 can force a visit to $s_6$ from the attractor $A = \Attr_0(s_6) = s_5 \vee s_6$.
Player 1 can force $A$ from $B = \Attr_1(A) = s_4 \vee s_5 \vee s_6$.
But $r = \bot$, because player 1 can escape to $s_1$.

So, a nested game is constructed over $s_0 \vee s_1 \vee s_2 \vee s_3 \vee s_4$, with player 1 wanting $\eventually (s_5 \vee s_6)$.
In the nested game, $A = \Attr_1(s_5 \vee s_6) = s_4 \vee s_5 \vee s_6$.
The attractor $B = \Attr_0(s_4 \vee s_5 \vee s_6) = \top$, and player 0 can keep player 1 in there, until player 0 visits $s_4 \vee s_5 \vee s_6$.
So, in the nested game, player 1 makes the assumption $\always \eventually \big(
	(s_0 \vee s_1 \vee s_2 \vee s_3) \rightarrow (s_4 \vee s_5 \vee s_6) \big)
	=
	\always \eventually \neg (s_0 \vee s_1 \vee s_2 \vee s_3)$.
This covers the cooperative winning set, in this example the entire game graph.

In implementation, the players need to communicate, and select a leader in a cyclic order.
Each player becomes a leader in turn.
Each time a player becomes a leader, it selects its next recurrence goal, in cyclic order.
It announces the current goal, by using an auxiliary integer variable dedicated to this purpose.
Note that this operation is analogous to centralized transducer construction \cite{Bloem12jcss}.
The goal corresponds to a game stack, as constructed above.
Therefore, all players switch to playing the game that corresponds to the current node (i.e., current state).
By construction of the stack, the play will be led to the selected goal.
When the goal is reached, the leader selects the next leader, and the sequence repeats.

\paragraph{Acknowledgments}

This work was supported in part by the TerraSwarm Research Center, one of six centers supported by the STARnet phase of the Focus Center Research Program (FCRP) a Semiconductor Research Corporation program sponsored by MARCO and DARPA.

	\appendix
	\bibliographystyle{IEEEtran}
\bibliography{assumption}

\end{document}